\documentclass[11pt,reqno]{amsart}
\usepackage{a4wide,mathptmx,hyperref}

\providecommand{\BBb}[1]{{\mathbb{#1}}}

\providecommand{\cal}[1]{{\mathcal{#1}}}   

\newcommand{\ang}[1]{\langle#1\rangle}

\newcommand{\Bcirc}{\overset{\lower 1.5pt%
              \hbox{$@,@,@,@,@,\scriptscriptstyle\circ$}}B{}}
\newcommand{\Binfty}{\overset{\lower 1.5pt%
              \hbox{$@,@,@,@,@,\scriptscriptstyle\infty$}}B{}}
\newcommand{\bigdot}{\mathbin{\raise.65\jot\hbox{$\scriptscriptstyle\bullet$}}}

\newcommand{\B}{{\BBb B}}
\newcommand{\C}{{\BBb C}}
\newcommand{\Cn}{{\BBb C}^n}

\newcommand{\ch}{\operatorname{ch}}

\newcommand{\coker}{\operatorname{coker}}

\newcommand{\disc}{\operatorname{disc}}
\newcommand{\dual}[2]{\langle\,#1,\,#2\,\rangle}

\newcommand{\dv}{\operatorname{div}}
\newcommand{\erd}{\overset{\lower 1pt\hbox{\large.}}{e}
                  \overset{\lower 1pt\hbox{\large.}}{r}}

\newcommand{\ess}{\operatorname{ess}}

\newcommand{\Fcirc}{\overset{\lower 1.5pt%
               \hbox{$@,@,@,@,@,\scriptscriptstyle\circ$}}F{}}
\newcommand{\fracc}[2]{{
                \textstyle\frac{#1}{\raise 1pt\hbox{$\scriptstyle #2$}}}}

\newcommand{\fracci}[2]{{\frac{#1}{\raise 1pt\hbox{$\scriptscriptstyle #2$}}}}

\newcommand{\lap}{\operatorname{\Delta}}
\newcommand{\loc}{\operatorname{loc}}

\newcommand{\mlap}{-\!\operatorname{\Delta}}

\newcommand{\nrm}[2]{\|#1\|_{#2}}

\newcommand{\norm}[2]{\mathinner{\|}#1\,|#2\|}

\newcommand{\op}[1]{\operatorname{#1}}

\newcommand{\N}{\BBb N}

\renewcommand{\Re}{\operatorname{Re}}

\newcommand{\R}{{\BBb R}}
\newcommand{\Rn}{{\BBb R}^{n}}

{
\end{minipage}%
\par\medskip\noindent}%
\newcounter{enmcount}\renewcommand{\theenmcount}{{\rm\arabic{enmcount}}}
\renewenvironment{enumerate}{%
\begin{list}{{\llap{\rm(\theenmcount)}}}{\setlength{\labelwidth}{\leftmargin}%
\usecounter{enmcount}}}{\end{list}}
\newcounter{rmcount}\renewcommand{\thermcount}{{\rm\roman{rmcount}}}
\newenvironment{rmlist}{%
\begin{list}{{\rm(\thermcount)}}{\setlength{\labelwidth}{\leftmargin}%
\usecounter{rmcount}}}{\end{list}}
\newcounter{Rmcount}\renewcommand{\theRmcount}{{\rm\Roman{Rmcount}}}
\newenvironment{Rmlist}{%
\begin{list}{{\rm(\theRmcount)}}{\setlength{\labelwidth}{\leftmargin}%
\usecounter{Rmcount}}}{\end{list}}

\newcommand{\scal}[2]{(\,#1\,|\, #2\,)}

\newcommand{\Set}[2]{\bigl\{\,#1\bigm| #2\,\bigr\}}

\newcommand{\supp}{\operatorname{supp}}

\renewcommand{\hat}[1]{\overset{{\scriptscriptstyle \wedge}}{#1}}

\numberwithin{equation}{section}

\newtheorem{thm}{Theorem}
\numberwithin{thm}{section}
\newtheorem{prop}[thm]{Proposition}
\newtheorem{lem}[thm]{Lemma}
\newtheorem{cor}[thm]{Corollary}

\theoremstyle{definition}

 \numberwithin{exercise}{section}
 
\theoremstyle{remark}
\newtheorem{rem}[thm]{Remark}

\title[Spectral properties of Witten-Laplacians]{{\sloppy
On the spectral properties of Witten-Laplacians,\linebreak[4] 
their range projections and\linebreak[4]
Brascamp--Lieb's inequality}} 
\author[Jon Johnsen]{{Jon Johnsen}$^*$}
\address{
Department of Mathematical Sciences, Aalborg University
 Fredrik~Bajers Vej 7E,\newline
 DK--9220~Aalborg~O, Denmark}
\email{jjohnsen@math.auc.dk}
\thanks{$^*$Supported by TMR grant FMRX-CT960001 of 
the European Commision,
`PDE and QM' at Universit\'e de Paris-Sud, France; partly by the Danish
Natural Sciences Research Council, grant 9700987.
\\[2\baselineskip]{\tt Appeared in Integral equations and operator theory, {\bf 36\/} (2000), 288--324}}
\subjclass{82B26; secondary 47G30, 47N30}

\begin{document}
\sloppy 

\newcommand{\A}{A_0}
\renewcommand{\AA}{A_1}
\newcommand{\Ak}{A_k}
\newcommand{\cov}{\op{cov}}
\newcommand{\var}{\op{var}}
\newcommand{\W}{\lap^{(0)}_{\Phi}}
\newcommand{\WW}{\lap^{(1)}_{\Phi}}
\newcommand{\Wk}{\lap^{(k)}_{\Phi}}

\begin{abstract}
  A study is made of a recent integral identity of B.~Helffer and
   J.~Sj{\"o}strand, which for a not yet fully determined class of  probability 
   measures yields a formula for the covariance of two functions (of
   a stochastic variable); in comparison with the
     Brascamp--Lieb inequality, this formula is a more flexible and in some contexts
     stronger means for the analysis of correlation asymptotics in statistical
     mechanics. Using a fine version of the
     Closed Range Theorem, the identity's validity is shown to be equivalent
   to some explicitly given 
     spectral properties of Witten--Laplacians on Euclidean space, and
     the formula is moreover deduced from the obtained abstract expression
     for the range projection. As a corollary, a generalised version of
     Brascamp--Lieb's inequality is obtained. For a certain class of
   measures occuring in statistical mechanics, explicit criteria for the
     Witten-Laplacians are found from the Persson--Agmon formula, from
     compactness of embeddings and from the
     Weyl calculus, which give results for closed range, strict positivity,
     essential self-adjointness and domain characterisations. 
\end{abstract}

\maketitle


\baselineskip=1.3\baselineskip



\section{Introduction and Main Results}  \label{intr-sect}
\enlargethispage{3\baselineskip}\thispagestyle{empty}

\subsection{Background} 
In 1976, H.~J.~Brascamp and E.~H.~Lieb \cite{BrLi76} proved the following
inequality for an arbitrary function $f$ in $C^1(\Rn)\cap L^2(\mu)$, when
the given measure $d\mu=e^{-\Phi}\,dx$ has a real-valued, \emph{strictly
convex} $C^2$ `potential' $\Phi$ with Hessian $\Phi''=(\partial^2_{jk}\Phi)_{j,k}$:  
\begin{equation}
  \int_{\Rn} |f(x)-\ang{f}|^2e^{-\Phi(x)}\,dx
  \le \int_{\Rn} \bigl(\nabla f(x)^T\Phi''(x)^{-1} \overline{\nabla f(x)}\bigr)
  e^{-\Phi(x)} \,dx;
  \label{Brascamp--Lieb-ineq}
\end{equation}
the measure $\mu$ is finite and may be normalised to $\int d\mu=1$ without loss of
generality (by adding $\log\int_{\Rn}d\mu$ to $\Phi$) which is done tacitly
throughout, so $\ang{f}:=\int fe^{-\Phi}\,dx$ equals $f$'s mean.

Since then this inequality has been used in physics, where the strict
convexity assumption on $\Phi$ in some contexts is a serious restriction;
e.g.\ this is the case for the analysis of asymptotics of correlations in 
statistical mechanics.

As another technique for such problems, B.~Helffer and J.~Sj{\"o}strand 
have recently  introduced an exact formula
\cite{HeSj93,Sj96,Hel96,Hel97'}
for the \emph{covariance} of two functions $g_1$, $g_2$ in 
$L^2(\mu)$, i.e.\ for
$\cov(g_1,g_2):=
\int_{\Rn}(g_1-\ang{g_1})\overline{(g_2-\ang{g_2})}e^{-\Phi}\,dx$
(in comparison the variance of $f$ enters in \eqref{Brascamp--Lieb-ineq}). 
Denoting the inner product of both $L^2(\Rn,\mu)$ and $L^2(\Rn,\mu,\Cn)$ by
$\scal{\cdot}{\cdot}_{\mu}$ for simplicity's sake, their identity may be
written as follows:
\begin{equation}
  \scal{g_1-\ang{g_1}}{g_2-\ang{g_2}}_{\mu}
  =\scal{\AA^{-1}\nabla g_1}{\nabla g_2}_{\mu}.
  \label{Brascamp--Lieb-id}
\end{equation}
This uses two elliptic differential operators $\A\ge0$ and $\AA\ge0$ on $\Rn$
(although $\A$ does not appear explicitly in \eqref{Brascamp--Lieb-id}); these
are equivalent, as observed in \cite{Sj96}, to Witten's Laplacians
\cite{Wit82}, and they have the expressions
\begin{equation}
  \A =\mlap+\Phi'\cdot\nabla,  \qquad
  \AA =\A\otimes I+\Phi''  .
  \label{AAA-eq}
\end{equation}
In the proofs of Helffer and Sj{\"o}strand, $\Phi$ has had a rather specific
nature, e.g.~with $\Phi''_{jk}$ bounded and $x\cdot\Phi'\ge C|x|^{1+\omega}$
in \cite{Sj96}; the $\Phi$ treated in \cite{Hel96} is essentially
$|x|^4-|x|^2$, see Example~\ref{exmp-ssect} below.  
Formula \eqref{Brascamp--Lieb-id} has also been used by 
A.~Naddaf and T.~Spencer \cite{NaSp96}, V.~Bach, T.~Jecko and
J.~Sj{\"o}strand \cite{BaJeSj97}, B.~Helffer~\cite{Hel97''} and others. 
Indirectly $\A$, $\AA$ appeared earlier in \cite{Sj93ii,HeSj93,Hel93}. 
 
Concerning formula \eqref{Brascamp--Lieb-id}, it should be noted that
$g_j$ only enters in the covariance through $Pg_j:=g_j-\int
g_j\,d\mu$, which is the orthogonal projection onto $L^2(\mu)\ominus\C$,
i.e.\ onto the complement of the constant functions. Because the
gradient provides another means to remove the part of $g_j$ lying in $\C$, 
it is natural to have $\nabla g_1$ and $\nabla g_2$ on the right hand side of
\eqref{Brascamp--Lieb-id} and to have an inverse of a second order differential
operator, like $\AA^{-1}$, to counteract the gradients. 

Viewed thus, \eqref{Brascamp--Lieb-id} may seem plausible, and this article
presents a study of it, resulting first of all in more general
sufficient conditions for the formula; secondly, conditions that are
equivalent to \eqref{Brascamp--Lieb-id} are given at an abstract level
(although these two improvements are formally substantial, the consequences for
statistical mechanics are to be investigated). 
Thirdly a more systematic and streamlined
approach to \eqref{Brascamp--Lieb-id} is presented. 

\begin{rem}
For the general importance of formula \eqref{Brascamp--Lieb-id}, 
recall that for a stochastic variable $X\colon \Omega\to \Rn$ 
with distribution $\mu$ on $\Rn$, the left hand side of
\eqref{Brascamp--Lieb-id} equals $\cov(g_1(X),g_2(X))$. So when $\mu$ is of
the type treated here, then \eqref{Brascamp--Lieb-id} provides a formula also for
such covariances. However, this is clear, and hence this consequence
shall not be treated below. 
\end{rem}

\subsection{Summary}
In this paper various\,---\,abstract and explicit\,---\,conditions are deduced
for \eqref{Brascamp--Lieb-id}. To give an application of these,
\eqref{Brascamp--Lieb-ineq} is derived in the general strictly convex case
from \eqref{Brascamp--Lieb-id} and moreover extended to the case $f\in
H^1_{\loc}(\Rn)\cap L^2(\mu)$; this supplements an explanation of B.~Helffer
of how \eqref{Brascamp--Lieb-id} implies the validity of
\eqref{Brascamp--Lieb-ineq} for certain $f$ when $\Phi$ is uniformly
strictly convex \cite{Hel96}.  

In the applications estimates like $\AA\ge c_0>0$ would clearly give a
bounded inverse $\AA^{-1}$, so that \eqref{Brascamp--Lieb-id} would yield
control of the covariance on the left hand side there. However, it turns out that
already \eqref{Brascamp--Lieb-id} \emph{itself} follows from such
a strict positivity.  

This fact is observed here (while further properties of the $A_j$ entered in 
\cite{Sj96, Hel96}) and proved for rather more general $\Phi$
than those considered hitherto. Actually, by means of
Hilbert space methods (the Closed Range Theorem), even weaker sufficient conditions
on $\A$ or $\AA$ (such as closed range of $\A$) are proved to be equivalent to
\eqref{Brascamp--Lieb-id}; this analysis is furthermore valid for arbitrary
probability measures $\mu$ on $\Rn$.  

These techniques are also applied to the associated $d$-complex in
$\mu$-weighted spaces and to the associated $A_k$ on forms of higher
degrees, see \eqref{dmu-cmplx} and \eqref{Ak-def} below. This is done 
because these objects may be of interest in statistical mechanics, and 
because the cases $k=0$ and $1$ have to be treated anyway in order to settle
the relations between the various conditions found for $\A$ and $\AA$.

Explicit criteria in terms of $\Phi$ are given partly by 
means of compact embeddings or bounds of essential spectra; partly
by a pseudo-differential treatment of the Witten--Laplacians on functions and
$1$-forms on $\Rn$, i.e.~of
\begin{equation}
\smash{
  \W=\mlap+\tfrac{1}{4}|\Phi'|^2-\tfrac{1}{2}\lap\Phi,
  \qquad
  \WW=\W
} \otimes I+\Phi'';   
  \label{WWW-id}
\end{equation}
these are unitarily equivalent to $\A$ and $\AA$. This approach consists in
a study of the $\Wk$ by means of 
the Weyl calculus of H{\"o}rmander \cite[18.4--6]{H}, and it
circumvents the difficulty that pseudo-differential techniques do not play
well together with the \emph{weighted} Sobolev spaces in which the $\Ak$ act. In
this analysis,
sufficient conditions for closed range of $\A$, strict positivity of $\AA$
and essential self-adjointness as well as characterisations of the maximal
domains of $\A$ and $\AA$ are established (through similar results for the $\Wk$).

When combined, the ps.d.\ and Hilbert space analyses show
formula \eqref{Brascamp--Lieb-id} for a class of 
$\Phi$ containing e.g.\ all polynomials of even degree $r\ge 2$, which
have positive definite part of degree $r$ (a change of the constant term
will renormalise to $\int d\mu=1$); in comparison the polynomials belonging
to the class of \cite{Sj96} have $r=2$ (a rather simpler
case because the term $\Phi''$ in \eqref{WWW-id} is bounded on
$L^2(\Rn,\Cn)$), while \cite{Hel96} covered cases with $r=4$.

\subsection{The main results}
In formula \eqref{Brascamp--Lieb-id} the probability measure $\mu$ will be
arbitrary to begin with, while 
$g_j$ will be considered either in the weighted Sobolev space
$H^1(\mu)$ equal to $\{\,u\in
L^2(\mu)\mid \forall j=1,\dots,n\colon \partial_ju\in L^2(\mu)\,\}$, whereby
$L^2(\mu):=L^2(\Rn,\mu,\C)$, or in $L^2(\mu)$ when this is
justified. When the measure $\mu$ is such that
$d\mu=e^{-\Phi(x)}\,dx$, then it will throughout be assumed that
\begin{equation}
  \int_{\Rn} e^{-\Phi(x)\,}dx=1 \quad\text{and}\quad\Phi\in C^2(\Rn,\R).
  \label{phi-cnd}
\end{equation}
To simplify notation (cf.\ \eqref{Ak-def} ff.~ below),
differentials will hereafter be used instead of gradients, 
so $\AA$ will act on suitable 1-forms in $L^2(\Rn,\mu,\wedge^1\Cn)$. To
explicate the operators that appear as $d_0$ and $d^*_0$ in \eqref{dmu-cmplx}
below and onwards, 
note that when 1-forms are identified with vector functions~$v$, 
then $d_0 f$, $d_0^*v$
identify with $\nabla f$ and $(\Phi'-\nabla)\cdot v$, respectively.
Also for simplicity, $\scal{\cdot}{\cdot}_{\mu}$ will for any $k$ refer to the scalar
product in $L^2(\Rn,\mu,\wedge^k\Cn)$, i.e.\ the space of $k$-forms with
coefficients in $L^2(\mu)$; on this space the following norm is used: 
\begin{equation}
  \nrm{\sum_{j_1<\dots<j_k} f_{j_1\dots j_k}dz^{j_1}\wedge\dots\wedge
      dz^{j_k}}{\mu}= (\int_{\Rn} (\sum |f_{j_1\dots j_k}|^2)\,d\mu)^{1/2}.
  \label{L2kmu-eq}
\end{equation}

It is of course a central question how the operators $\A$ and $\AA$ are defined
precisely, but it will be equally important to make
sense of the inverse $\AA^{-1}$. For this reason, it is worthwhile to
define $\A$ and $\AA$ as `Hodge Laplacians' to begin with; this is not only a
most general approach (it works for arbitrary probability measures $\mu$),
but it also leads in a very natural way to a fruitful discussion, by simple operator
theoretical methods, of the invertibility of $\AA$. This will be explained
in the following, before the theorems are presented.

But first of all it should be emphasised that the present article focuses on
the following problem for the identity \eqref{Brascamp--Lieb-id}:
\begin{quote}
   Which probability measures $\mu$ have the property that formula \eqref{Brascamp--Lieb-id}
   holds for all functions $f$, $g$ in $H^1(\mu)$?
\end{quote}
Of course other problems would be equally meaningful (such as fixing two
functions $f$ and $g$ and then search for the $\mu$ for which
\eqref{Brascamp--Lieb-id} would be true); but for simplicity's sake the
discussion will here be restricted to the above-mentioned problem.

Secondly, as the point of departure it is useful to adopt the following definitions of
$\A$ and $\AA$ as the Hodge Laplacians%
\footnote{For $d\mu=e^{-\Phi}\,dx$, an interpretation of $e^{-\Phi}$ as
$\sqrt{\det(g_{ij})}$ would for any Riemannian metric $(g_{ij})$ lead to 
a Laplace--Beltrami operator which resembles $\A$,
$\AA$,\dots but differs; indeed, in this way already \eqref{L2kmu-eq} would be
unconventional in the study of the de Rham complex;
see e.g.~\cite[2.1]{Jost98}.}
\begin{equation}
  \Ak^{(\op{H})}=d^*_kd_k+d_{k-1}d^*_{k-1},\qquad 
  D(\Ak^{(\op{H})})=D(d^*_kd_k)\cap D(d_{k-1}d^*_{k-1}),
  \label{Ak-def}
\end{equation} 
associated to the complex (where $d_{-1}\equiv 0$)
\begin{equation}
  (0)\longrightarrow L^2(\mu)\overset{d_0}{\longrightarrow}
  L^2(\mu,\wedge^1\Cn)\overset{d_1}{\longrightarrow}\dots
  \overset{d_{k-1}}{\longrightarrow}
  L^2(\mu,\wedge^k\Cn)\overset{d_k}{\longrightarrow}\dots.
  \label{dmu-cmplx}
\end{equation}
Here $d_k$ denotes the exterior differential of $k$-forms; this is in
general a first order differential operator acting in the distribution
sense, but in the context above, $d_k$ is equipped with its maximal domain
as an unbounded, closed, densely defined operator from $L^2(\mu,\wedge^k\Cn)$ to
$L^2(\mu,\wedge^{k+1}\Cn)$; clearly $D(d_0)=H^1(\mu)$. 
This is an example of a Hilbert complex in the sense of J.~Br{\"u}ning and
M.~Lesch \cite{BrLe92}, where such complexes are described
in a clear way (by comparison the present article focuses on the
conditions implying that \eqref{dmu-cmplx} is a Fredholm complex rather than
the conclusions that would follow from this property).

For $d_k$ the closure of the range is denoted by $X_{k+1}$,
\begin{equation}
  X_{k+1}:=\overline{R(d_k)}
  \label{X-id}
\end{equation}
and the kernel by $Z(d_k)$; as a convention $X_0:=\C$. The Hilbert space adjoint, $d_k^*$,
also enters in the  Hodge Laplacians, cf.~\eqref{Ak-def}; for simplicity the
superscript `(H)' is suppressed in the sequel, for the definition in
\eqref{Ak-def} will be in effect until Section~\ref{Ak-sect}, unless
otherwise is explicitly stated.

Thirdly, a series of small remarks will clarify the situation:
using the orthogonal projection $P$ onto $L^2(\mu)\ominus \C$, formula
\eqref{Brascamp--Lieb-id} may be written as
\begin{equation}
  \scal{Pg_1}{Pg_2}_{\mu}=\scal{\AA^{-1}d_0g_1}{d_0g_2}_{\mu}.
  \label{PPAA-eq}
\end{equation} 
Moreover, $d_0g_1$ and $d_0g_2$ belong to $X:=X_1$, and since
$d\circ d\equiv 0$ implies $X_k\subset Z(d_{k+1})$, the 
subspace $X$ of $L^2(\mu,\wedge^1\Cn)$ is invariant under $\AA$ by
\eqref{Ak-def}; in fact
\begin{equation}
  \AA|_X=(d_0d_0^*)|_X.
\end{equation}
Hence one has the following identity of unbounded operators in $X$:
\begin{equation}
  (\AA^{-1})|_X= (d_0d_0^*|_X)^{-1};
  \label{AAinv-eq}
\end{equation}
in particular, $D(\AA^{-1})\cap X$ equals the range $R(d_0d_0^*|_X)=d_0d_0^*(X)$.

It is now straightforward to verify the implications
\begin{align}
  &  \text{\eqref{PPAA-eq} holds for $g_1$, $g_2\in H^1(\mu)$}
\\
  \iff &
   Pf= d_0^*(d_0d_0^*)^{-1}d_0 f\quad\text{holds for $f\in H^1(\mu)$}.
  \label{P-id}
\end{align}
Indeed, these properties are, by \eqref{AAinv-eq} and the self-adjointness of $P$,
both equivalent to 
\begin{equation}
  \scal{(d_0d_0^*)^{-1}d_0f}{d_0g}_\mu=\scal{Pf}{g}_\mu \quad\text{holds for
  all}\quad f,g\in H^1(\mu);
  \label{P-id'}
\end{equation}
e.g.\ it is found when \eqref{P-id} holds that $(d_0d_0^*)^{-1}d_0f$ is an
element of $D(d_0^*)$, whence \eqref{P-id'} follows from the definition of
the adjoint. The other implications follow in a similar manner.

Therefore the formulated problem for \eqref{PPAA-eq} has been
reduced to the just given property in \eqref{P-id} for $P$,
and this rewriting as a `linear' problem makes the analysis more
straightforward, as we shall see immediately.

\begin{lem}
$Pf= d_0^*(d_0d_0^*)^{-1}d_0 f$ holds whenever $f-\ang{f}$ belongs to 
$D(d_0)\cap R(d_0^*)$.
\end{lem}
\begin{proof}
When $f-\ang{f}$ is in $H^1(\mu)\cap R(d_0^*)$, then $d_0f=d_0(d_0^* v)$
holds for some
$v\in X\cap D(d_0^*)$, because $X^\perp= Z(d_0^*)$. Hence $d_0f$ belongs to
$D(d_0d_0^*|_X^{-1})$, and so
\begin{equation}
  d_0^*(d_0d_0^*|_X)^{-1}d_0 f= d_0^*\cdot I v= f-\ang{f},
\end{equation}
since $d_0d_0^*|_X^{-1}$ maps into
$D(d_0d_0^*)\subset D(d_0^*)$.
\end{proof}

The usefulness of the lemma in connection with the proof of
\eqref{Brascamp--Lieb-id} was independently discovered by V.~Bach, T.~Jecko
and J.~Sj{\"o}strand \cite[(II.15)]{BaJeSj97}. 

Conversely \eqref{P-id} implies that $f-\ang{f}=Pf$ belongs to $D(d_0)\cap
R(d_0^*)$ for any $f\in H^1(\mu)$. However, it is more useful
to ask the following question: which probability mesures $\mu$ have the
property that 
\begin{equation}
  \forall f\in H^1(\mu)\colon f-\ang{f}\in D(d_0)\cap R(d_0^*)\quad\text{?}
\end{equation}
Since $H^1(\mu)=D(d_0)$, it is clear that $\mu$ has this property if and
only if 
\begin{equation}
  H^1(\mu)\ominus \C \subset R(d_0^*).
  \label{Rd*1-eq}
\end{equation}
(Here and throughout $F\ominus\C$ stands for the elements of a given
subspace $F\subset L^2$ which are orthogonal to the constant functions.)
The inclusion \eqref{Rd*1-eq} would  obviously be true if $\mu$ is such that 
\begin{equation}
  R(d_0^*)= L^2(\mu)\ominus\C;
  \label{Rd*-eq}
\end{equation}
since $\overline{R(d_0^*)}= Z(d_0)^\perp=\C^\perp$, condition \eqref{Rd*-eq} is
equivalent to closedness of the range $R(d_0^*)$.

Because \eqref{Rd*-eq} also expresses that $R(d_0^*)$ should be the maximal
possible subspace of $L^2(\mu)$, one could in view of \eqref{Rd*1-eq}
envisage that \eqref{Rd*-eq} is sufficient but unnecessary.

However, this is not the case, for a fine version of the Closed Range
Theorem yields that \eqref{PPAA-eq} is equivalent to the closedness of
$d_0^*$'s range, hence to \eqref{Rd*-eq}. This is explained  for arbitrary
unbounded operators in general Hilbert spaces in Section~\ref{CR-sect}
below, but the consequences for the $\Ak$ 
are summed up in the following theorem. In particular, 
 the problem formulated for \eqref{Brascamp--Lieb-id} above has been
rephrased by means of a set of equivalent conditions on $\A$, $\AA$, $d_0$
and $d_0^*$, that one finds for $k=0$ in 

\begin{thm}
  \label{kk1-thm}
Let $k$ be a fixed index in the complex
\eqref{dmu-cmplx}. Then closedness of each of the spaces $R(\Ak)$, $R(d_k)$
and $R(d^*_k)$ as well as strict positivity of $A_{k+1}|_{X_{k+1}}$ are
equivalent properties, and if $P_k$ is the
projection from $L^2(\mu,\wedge^k\Cn)$ onto $R(A_k)$, these properties are
also equivalent to the validity of
\begin{equation}
  P_k=d^*_k (A_{k+1}|_{X_{k+1}})^{-1} d_k\quad\text{on all of}\quad D(d_k) .
\end{equation}

In the affirmative case, $A_k|_{Z(d_k)^\perp}=d^*_kd_k$ and has
closed range $R(A_k)=R(d^*_k)$ and $A_{k+1}|_{X_{k+1}}=d_kd^*_k$.
Moreover, there is an equivalence
\begin{equation}
  A_k|_{Z(d_k)^\perp}=U^*A_{k+1}|_{X_{k+1}}U
  \label{AAuni-eq}
\end{equation}
with $U$ equal to the unitary $d_k(d_k^*d_k|_{X_k})^{-1/2}$ from
$L^2(\mu,\wedge^k\Cn)\ominus Z(d_k)$ to $X_{k+1}$, and
consequently 
\begin{equation}
  \sigma(A_k|_{Z(d_k)^\perp})= \sigma(A_{k+1}|_{X_{k+1}})
\end{equation}
with the analogous relation for the essential spectra.
\end{thm}

The point of the proof of this result is to combine the usual
estimates from below of the adjoint (here $d^*_0$) with the fact that $T^*T$ is 
self-adjoint for any densely defined, closed operator $T$. In the abstract set-up, the
range projection of $T^*T$ has the form $P=T^*(TT^*|_{R(T)})^{-1}T$;
when applied to $\A$ and $\AA$, this yields
\eqref{P-id} and thus \eqref{Brascamp--Lieb-id} at least for $u$, $g_1$ and
$g_2\in H^{1}(\mu)$.

However, neither closedness of $R(\A)$, $R(d_0)$, $R(d_0^*)$ nor positivity of
$\AA|_X$ are easy to analyse when the $\Ak$ are defined from
\eqref{Ak-def} for an arbitrary probability measure $\mu$.
From Section~\ref{Ak-sect} below and onwards, we shall therefore work under the assumption
that $\mu$ has a density $e^{-\Phi(x)}$ with respect to Lebesgue measure for
some $\Phi$ fulfilling \eqref{phi-cnd}.

Using this assumption, an alternative variational
definition of the $\Ak$ (i.e.\ by means of sesqui-linear forms, or Lax--Milgram's
lemma) is introduced in Section~\ref{Ak-sect} below; thereafter it is seen
that this yields the Friedrichs extension from
$C^\infty_0(\Rn,\wedge^k\Cn)$, 
and it is then shown, for the variationally defined operators, that \eqref{Ak-def}
holds both in the distribution sense and as a formula for unbounded operators. 
(Concerning essential self-adjointness of the $\Ak$, see the remarks in 
Section~\ref{final-sect} below.)

Using the definition by Lax--Milgram's lemma,
it is proved in Section~\ref{d-sect} below that the regularity
assumption on the $g_1$, $g_2$ and
$f$ above may be relaxed from $H^1(\mu)$ to  $L^2(\mu)$:

\begin{thm}
  \label{Brascamp--Lieb-thm}
Let $\Phi$ satisfy \eqref{phi-cnd}, and suppose
that $\A$ as an unbounded operator in $L^2(\mu)$ has closed range,
$R(\A)=\overline{R(\A)}$.
Then \eqref{Brascamp--Lieb-id} holds for all $g_1$ and $g_2$ in $L^2(\mu)$.

Moreover, it then holds that $L^2(\mu)=R(\A)\oplus \C$, and for every $u\in L^2(\mu)$,
\begin{equation}
  Pu=d_0^*\AA^{-1}d_0 u 
  \label{P-eq}
\end{equation}
when $Pu=u-\int u\,d\mu$ denotes the orthogonal projection onto $L^2(\mu)\ominus \C$.
\end{thm}

Note also that when $g_2$ is in $L^2(\mu)\setminus H^1(\mu)$, it is
understood in Theorem~\ref{Brascamp--Lieb-thm} that
the right hand side of \eqref{Brascamp--Lieb-id} should be read as a
duality $\dual{\tilde{\cal A}_1^{-1}dg_1}{d\overline{g_2}}_{\tilde V\times \tilde
V'}$ for a certain Hilbert space $\tilde V$ with isomorphism $\tilde{\cal A}_1\colon \tilde V\to \tilde V'$ onto its dual. See
Section~\ref{prf-ssect} for details.

Although the identification of $\AA$, or rather $\AA|_{X}$, with a
restriction of $\cal A_1$ is a well-known procedure for operators defined by
the Lax--Milgram lemma, it is really the validity of a certain Poincar\'e
inequality for $\tilde V$ which makes this possible. Because $\tilde V$ is
the form domain of $\AA|_X$, one may of course ask directly whether such
Poincar\'e inequalities for the form domain of $\AA$ or $\AA|_X$ would imply
\eqref{Brascamp--Lieb-id}; moreover, also the exactness of the complex
\eqref{dmu-cmplx} may be considered. However, that the Poincar\'e
inequalities are equivalent to e.g.\ positivity of $\A$ and $\AA|_X$,
respectively, while exactness is an `intermediate' condition, is  proved in
Theorem~\ref{ran-thm} below, where the full interrelationship among the
various criteria is settled.

By exploiting the possibility in
Theorem~\ref{Brascamp--Lieb-thm} of taking the $g_j\in L^2(\mu)$, one finds
as an application that \eqref{Brascamp--Lieb-id} implies a generalisation of
the inequality in \eqref{Brascamp--Lieb-ineq} from $C^1$ to $H^1_{\loc}(\mu)$: 

\begin{cor}
  \label{Brascamp--Lieb-cor}
Let $\Phi\in C^2(\Rn,\R)$ satisfy $\int e^{-\Phi}\,dx=1$ and be
strictly convex, i.e.\
\begin{equation}
  \sum\Phi''_{jk}(x)z_j\overline{z_k}>0\quad\text{for all}\quad
  x\in\Rn,\quad\text{all } z\in\C^n\setminus\{0\}. 
\end{equation}
Then the inequality \eqref{Brascamp--Lieb-ineq} holds for all $f\in
L^2(\Rn,\mu)\cap H^1_{\loc}(\Rn)$.  
\end{cor}

This generalises \cite{Hel96}, where B.~Helffer sketched how
\eqref{Brascamp--Lieb-id} implies \eqref{Brascamp--Lieb-ineq} for
{\em uniformly\/} strictly convex $\Phi$, when in addition $\nabla f$ is
bounded.  The general case was left open there, although with indications that
\eqref{Brascamp--Lieb-id}, because it is an exact formula, is more
powerful. (The reference to \eqref{Brascamp--Lieb-id} refines an explanation
from 1993, see \cite{Hel93}, where $\A$ and $\AA$ were introduced for
\eqref{Brascamp--Lieb-ineq} without \eqref{Brascamp--Lieb-id}.)

For the sake of the proof, it should be recalled from \cite{Hel96} that in
the uniformly strictly 
convex case, say $\Phi''(x)\ge c_0>0$ on $\Rn$ for every $x$, the idea
behind \eqref{Brascamp--Lieb-id}$\implies$\eqref{Brascamp--Lieb-ineq} is to
infer from the formal expression in \eqref{AAA-eq} that 
\begin{equation}
  \AA\ge \Phi''\quad\text{and therefore}\quad \AA^{-1}\le (\Phi'')^{-1}\le
\tfrac{1}{c_0}. 
\end{equation}
This actually requires justification, for inclusions between the
domains $D(\AA)$ and $D(\Phi'')$ need not exist under the assumption
\eqref{phi-cnd}. Nevertheless one may envisage that $\AA\ge c_0>0$
then holds\,---\,so that
$\AA^{-1}$ would be well defined\,---\,and this is verified in
Section~\ref{Brascamp--Lieb-ssect} below. However, the proof of 
Corollary~\ref{Brascamp--Lieb-cor} is first completed in
Section~\ref{prf-ssect}, after that of Theorem~\ref{Brascamp--Lieb-thm}.

\bigskip

The identity \eqref{Brascamp--Lieb-id} is established independently of
compactness of 
$H^1(\mu)\hookrightarrow L^2(\mu)$, so in particular the
essential spectra of $\A$  and $\AA$ need not be empty.
In this connection some simple sufficient conditions are given in
Section~\ref{inv-sect}, exploiting the Persson--Agmon formula and results
due to P.~Bolley, M.~Dauge and B.~Helffer \cite{BoDaHe89}.

However, slightly stronger conditions allow 
$\A$ and $\AA$ to be analysed by means of the
Weyl calculus \cite[18.4--6]{H}, leading to closed range, essential
self-adjointness and positivity. With $v^\alpha:=v_1^{\alpha_1}\cdot\dots\cdot
v_n^{\alpha_n}$ for $\alpha\in\N_0^n$ and $v\in\C^n$, or $v=-i\nabla=:D$,
the following results are restatements of
Theorems~\ref{psdCR-thm}, \ref{psdAA-thm} and \ref{AApos-thm} below:

\begin{thm}
  \label{AAA-thm}
Let $\Phi(x)\in C^\infty(\Rn,\R)$  satisfy $\int
e^{-\Phi(x)}\,dx=1$; suppose also that $|\Phi'(x)|\to\infty$ for
$|x|\to\infty$, that $|D^\beta \Phi(x)|\le C_\beta
(1+|\Phi'(x)|^2)^{\frac{1}{2}}$ for
every multi\-index $|\beta|\ge1$ and that, for some $M$, any $D^\beta\Phi$
with $|\beta|=M$ is bounded on~$\Rn$. 

Then $\A$ has closed range in
$L^2(\Rn,\mu,\C)$, so that the conclusions of Theorem~\ref{Brascamp--Lieb-thm}
are valid, and $\A$ is essentially self-adjoint (when considered) on $C^\infty_0(\Rn)$.
The operator $\AA$ is essentially self-adjoint on
$C^\infty_0(\Rn,\wedge^1\Cn)$ and has closed range.

Moreover, if there exist $\omega>0$, $C>0$ such that 
$x\cdot\Phi'(x)\ge |x|^{1+\omega}/C$ holds for all
$|x|\ge C$, then $\AA>0$ on $L^2(\Rn,\mu,\wedge^1\Cn)$.
\end{thm}

The considered class of $\Phi$ is larger than those in
\cite{Sj96,Hel96}; e.g.\ any polynomial $\Phi$ of even degree $\ge 2$
satisfies the assumptions when the part of highest degree is positive
definite; this includes the condition for strict positivity of $\AA$
(seen as in Example~\ref{exmp-ssect} below).

\begin{thm}
  \label{AAAdm-thm}
Under the hypothesis for $\Phi$ made in the beginning of
Theorem~\ref{AAA-thm}, a function 
$u$ belongs to $D(\A)$ precisely when it has the property 
\begin{equation}
  (\Phi')^\beta D^{\alpha}u\in L^2(\Rn,\mu) \quad\text{for all 
  $\alpha$ and $\beta$ such that}\quad  |\alpha+\beta|\le 2.
  \label{Adm-eq}
\end{equation}
A form $\sum_{l=1}^n v_l\,dx^l$ is in $D(\AA)$ if and only if each $v_l$ 
satisfies \eqref{Adm-eq}.

Consequently $\A$ and $\AA$ have compact resolvents,
hence their spectra are discrete.
\end{thm}
When the condition for $\AA>0$ is fulfilled for some
$\omega\ge 1$, then $D(\A)$ contains the set $B_2(\Rn,\mu)$, defined as those
$u$ for which $x^\beta D^\alpha u\in L^2(\mu)$ when $|\alpha+\beta|\le2$.

Theorem~\ref{AAAdm-thm} may be used for example in distortion
arguments for $\AA^{-1}$ in the correlation estimates; see 
\cite[Thm.~4.1]{Hel96}. Recently Wei-Min Wang \cite{Wan99} has  applied
Theorem~\ref{AAA-thm}, or rather the corresponding facts on $\WW$ in (the
proofs of) 
Theorems~\ref{psdAA-thm}--\ref{AApos-thm} below; indeed the explicit
conditions on $\Phi$ implying closed range and injectivity was
used in \cite[Sect.~3--4]{Wan99}.  In addition the unitary equivalence in
\eqref{AAuni-eq} has 
been used by V.~Bach, T.~Jecko and J.~Sj{\"o}strand \cite[Prop.~V1]{BaJeSj97}.

\begin{rem} 
For general probability measures $\mu$ it is questionable
which expressions the associated operators $\A$, $\AA$ can have: already for $\mu$
absolutely continuous with respect to Lebesgue measure, $\Phi$ equals $+\infty$ in 
$\Rn\setminus\supp\tilde\mu$, and this Borel set may be so irregular that
when one attempts a calculation of the $d_0^*$ in \eqref{Ak-def}, then
$\partial_je^{-\Phi}$ is unequal to $-\Phi'_je^{-\Phi}$ (and the latter may
even be undefined in $\cal D'$).
\end{rem}

\subsubsection*{Acknowledgement} I am very grateful to B.~Helffer for
several discussions on the subject and for letting me carry further some
ideas from \cite{Hel96} with this work, and I thank V.~Bach and T.~Jecko
for fruitful conversations and for explaining me about the applications to
statistical mechanics.


\section{Notation and Preliminairies}   \label{notation-sect}
For an operator $T$ in a Hilbert space with scalar product
$\scal{\cdot}{\cdot}$ and norm $\|\cdot\|$, the domain, range and
kernel is written $D(T)$, $R(T)$ and $Z(T)$, respectively, while 
$\rho(T)$ and $\sigma(T)$
stand for the resolvent and spectrum of $T$; the space of bounded operators is
denoted by $\B(H)$. The essential spectrum
$\sigma_{\ess}(T)=\sigma(T)\setminus \sigma_{\disc}(T)$ consists in
the self-adjoint case of those $\lambda\in\C$ for which there is 
a sequence with $\|x_n\|=1$ and $(T-\lambda)x_n\to
0$ while  $x_n\rightarrow 0$ weakly; it is the complement of the
$\lambda$ for which $T-\lambda$ is Fredholm $D(T)\to H$ in the graph
topology, i.e.\ has $Z(T-\lambda)$ of finite dimension and
$R(T-\lambda)$ closed (using \cite[Prop.~19.3]{H}).
Moreover, 
\begin{equation}
  m(T)=\inf\Set{\Re\scal{Tx}{x}}{\|x\|=1}
\end{equation}
is the lower bound of $T$.
Occasionally the norm  $\nrm{\cdot}{X}$ in a space $X$ is written
$\norm{\cdot}{X}$, to avoid unnecessary subscripts. 

Given a triple $(H,V,s)$ consisting of two Hilbert spaces $V\hookrightarrow H$
with bounded, \emph{dense} injection and a bounded sesqui-linear form $s(\cdot,\cdot)$
on $V$, then \emph{coerciveness}\,---\,i.e.\ existence of $c>0$ and $k\in\R$ such that 
\begin{equation}
  \Re s(u,u) \ge c\nrm{u}{V}^2- k\nrm{u}{H}^2 \quad\text{for all}\quad u\in V,
\end{equation}
gives the following for the operator $S$ defined on $\bigl\{\,u\in V \bigm| \exists f\in
H\forall v\in V: s(u,v)=\scal{f}{v}_{H}\,\bigr\}$ by the formula $Su=f$:

\begin{lem}
  \label{Lax--Milgram-lem}
$1^\circ$ $S$ is closed in $H$ with $D(S)$ dense in $V$. The adjoint $S^*$
is the analogous operator defined from $s^*(v,w):=\overline{s(w,v)}$.

$2^\circ$ When $s(\cdot,\cdot)$ is positive definite on $V$ (i.e.\
$s(v,v)>0$ for $v\ne 0$, e.g.\ for $k=0$),
then $S$ extends to the isometry $\cal A$ of $(V, s(\cdot,\cdot))$ onto
$V^*$, its antidual, and $D(S)=\cal A^{-1}(H)$.
\end{lem}

This result, known as Lax--Milgram's lemma, may be proved straightforwardly
in the fashion of J.--L.~Lions and E.~Magenes \cite[Sect.~9.1]{LM}.

$C^\infty_0(\Rn)$ denotes the space of infinitely differentiable
functions with compact support, and $\cal D'(\Rn)$ its dual; $\cal D'{}^k$
the subspace of distributions of order $k$;
$\dual{u}{\varphi}=u(\varphi)$ for all $u\in\cal D'$ and $\varphi\in
C^\infty_0$. For $L^2(\Rn,\mu)$ the scalar product and norm is written
$\scal{\cdot}{\cdot}_{\mu}$ and $\|\cdot\|_{\mu}$ respectively,
although with $\mu$ omitted in case of the Lebesgue measure. Similar
notation is adopted for the space of $k$-forms
$L^2(\Rn,\mu,\wedge^k\C)$; in general, for a space $F$ of functions
$\Rn\to\C$, the set $F(\Rn,\wedge^k\C^n)$ consists of the differential forms
with coefficients therein. 

A differential form of degree $k$ with complex $C^\infty$-coefficients has
the form
\begin{equation}
  f={\textstyle\sum_{|J|=k}'} f_J(x)\,dz^J;
\end{equation}
here $\sum'$ indicates summation over increasing $k$-tuples
$J=(j_1,\dots,j_k)$, i.e.\
strictly increasing maps $\{\,1,\dots,k \,\}\to \{\,1,\dots,n \,\}$; and
$dz^J:=dz^{j_1}\wedge\dots\wedge dz^{j_k}$ stands for the usual $k$-linear
anti-symmetric map $(\Cn)^k\to \C$ derived from the
Cartesian coordinates in $\Cn$. By anti-symmetry in the indices,
$f_{jL}=f_J\varepsilon^{jL}_{J}$ where $\varepsilon^{jL}_J=0$ unless $J$ has
the same elements as $jL:=(j,l_1,\dots,l_{k-1})$, in which case
$\varepsilon^{jL}_J$ is the sign of the permutation
$\left(\begin{smallmatrix} jL\\J\end{smallmatrix}\right)$.

For the distribution space
$\cal D'(\Rn,\wedge^k\Cn)$, see Appendix~\ref{curr-app}.

\section{An Operator Approach}  \label{CR-sect}
 It is shown in this section how Hilbert space methods can provide detailed
information about \eqref{Brascamp--Lieb-id}, using the rewriting given in
\eqref{P-id}.  The basic observation is that a similar projection appears in 
the Closed Range Theorem, at least in the version established below where
the self-adjointness of $T^*T$ and $TT^*$ is incorporated for this purpose. 

\bigskip

Let in the sequel $T\colon H\to H_1$ be a densely defined, closed operator between
Hilbert spaces $H$ and $H_1$, and let $F\subset H$ and $F_1\subset H_1$
denote two closed subspaces such that 
\begin{equation} 
  F=\overline{R(T^*)},\qquad  R(T)\subset F_1.
  \label{FF1-eq}
\end{equation}
Here the possibility of taking $F_1$ different from
$\overline{R(T)}$ is adopted from H{\"o}rmander's treatment of the
$\bar\partial$-complex \cite[Ch.~4]{Hdbar}; in an analogous way 
this is useful for the below study
of the exterior derivative $d$ in $L^2(\Rn,\mu,\wedge^k\Cn)$ spaces, where
$F_1=Z(d_{k+1})$ is a natural choice when $T=d_k$.

The closedness of $T$'s range is closely connected to the properties of the
operators
\begin{equation}
  S=T^*T|_{F},\qquad S_1=TT^*|_{F_1}
  \label{SS1-eq}
\end{equation}
and to the orthogonal projection $P$ onto $F$. In fact one has the next
result, which might be folklore, but nevertheless is formulated as a
theorem in view of the clarification it gives for \eqref{Brascamp--Lieb-id}:

\begin{thm}[Closed Range Theorem]
  \label{CR-thm}
 When $T$ is an operator as above, and the set-up in
 \eqref{FF1-eq}--\eqref{SS1-eq} is used, then the following properties are equivalent: 
\begin{rmlist}
  \item  \label{RT-cnd} $R(T)$ is closed and equal to $F_1$.
  \item  \label{lb-cnd} There exists $c>0$ such that $\norm{y}{H_1}\le
     c\norm{T^*y}{H}$ for $y\in D(T^*)\cap F_1$.
  \item  \label{RT*-cnd} $R(T^*)$ is closed and $F_1^\perp = Z(T^*)$.
  \item  \label{S1-cnd} $S_1$ is injective and has closed range. 
  \item   \label{S-cnd}  $S$ is injective and has closed range.
  \item   \label{prj-cnd} $S_1$ is injective and 
    \begin{equation}
        Px=T^*S_1^{-1}Tx\quad\text{for all}\quad x\in D(T).
       \label{proj-id}
    \end{equation}
\end{rmlist}
In the affirmative case, $S$ and $S_1$ are unitarily equivalent, that is
\begin{equation}
  S=U^*S_1 U
  \label{iseq-eq}
\end{equation}
holds for $U=TS^{-\frac{1}{2}}$, which is an isometry $U\in\B(F,F_1)$.
Consequently
$\sigma(S)=\sigma(S_1)$ and $\sigma_{\ess}(S)=\sigma_{\ess}(S_1)$. 
\end{thm}

\begin{proof} Note first that $F_1^\perp\subset Z(T^*)$ because of the assumption
on $F_1$, and that $D(T^*)$ is invariant under projection onto $F_1$ and
$F_1^\perp$; indeed, if $D(T^*)\ni y=f+f^\perp$ where $f\in F_1$ and $f^\perp\in
F_1^\perp$, then $f^\perp\in Z(T^*)$ and $f\in D(T^*)$. Moreover, $S_1$ is
densely defined in $F_1$: if $D(TT^*)\ni y_k\to y\in F_1$, one may split
$y_k=f_k+f^\perp_k$, where $f_k\in F_1\cap D(TT^*)=D(S_1)$ while
$\nrm{y-f_k}{}\le\nrm{y-y_k}{}\searrow 0$;  
the self-adjointness of $TT^*$ then carries over to $S_1$. 
Because the roles of $T$ and $T^*$ may be interchanged, also $S$ is
self-adjoint. 

\eqref{RT-cnd}$\iff$\eqref{lb-cnd} is proved in
\cite[Lem.~4.1.1]{Hdbar}. When \eqref{lb-cnd} holds, 
then $Z(T^*)\ni z=f+f^\perp$ (with $f^{(\perp)}$ as above) implies $T^*f=0$,
hence $f=0$, so $Z(T^*)=F_1^\perp$; hence $R(T^*)$ equals
$R(T^*|_{F_1\cap D(T^*)})$, and the latter is closed by \eqref{lb-cnd} so
\eqref{RT*-cnd} is obtained. To deduce \eqref{lb-cnd} from
\eqref{RT*-cnd}, it suffices to consider $T^*\colon 
F_1\cap D(T^*)\to H$ as a bounded operator in the graph norm and 
apply the Open Mapping Theorem.
Because \eqref{RT-cnd} and \eqref{RT*-cnd} are equivalent, so would
\eqref{S-cnd} and \eqref{S1-cnd} be once \eqref{RT-cnd}$\iff$\eqref{S1-cnd}
is proved.

From \eqref{RT*-cnd} injectivity follows since $TT^*z=0$ yields $z\in
Z(T^*)=F_1^\perp$; clearly $R(S_1)\subset R(T)$, but any $y\in R(T)$ equals $Tx$
for some $x\in D(T)\ominus Z(T)\subset Z(T)^\perp=R(T^*)$; thus $R(S_1)$
equals $R(T)$, hence is closed. This shows \eqref{S1-cnd}.

When \eqref{S1-cnd} holds, $S_1^{-1}$ is
defined on $F_1$ (because $R(S_1)=Z(S_1)=\{0\}$), and $S_1^{-1}\in \B(F_1)$ since 
$S_1$ is closed. $S_1^{-1}$ maps into $D(TT^*)$, so
for every $x\in D(T)$ it is obvious that $x\in D(T^*S_1^{-1}T)$ and that
\begin{equation}
  x-T^*S_1^{-1}Tx\in Z(T);
\end{equation}
since $1-P$ is the projection onto $Z(T)$, this entails \eqref{proj-id},
for
\begin{equation}
  0=P(x-T^*S_1^{-1}Tx)=Px-T^*S_1^{-1}Tx.
\end{equation}

Finally, \eqref{prj-cnd}$\implies$\eqref{RT-cnd}, for injectivity of
$S_1$ yields $F_1=\overline{R(T)}$, while \eqref{proj-id} shows that 
\begin{equation}
  R(T)\subset D(S_1^{-1})= R(S_1) \subset R(TT^*) \subset R(T);
\end{equation}
hence $R(T)=R(TT^*)$, and since this implies that $R(T)$ is
closed, \eqref{RT-cnd} is obtained.

\newcommand{\T}{\bar T}
However, for completeness' sake an elementary proof of the just mentioned
implication shall be given. When $R(T)=R(TT^*)$, then one can pass to a domain
consideration for the operator $\T=(T|_{F})^{-1}$ and use that
\begin{equation}
  \T=\T(I+\T^*\T)^{-1}+ \T\T^*\T(I+\T^*\T)^{-1}.
  \label{Tres-eq}
\end{equation}
Since $\T^*\T=(TT^*|_{F})^{-1}$, it is self-adjoint $\ge0$ in $F_1 $ with a square root
$Q=Q^*=(\T^*\T)^{1/2}\ge0$ fulfilling
\begin{equation}
  D(Q)= D(Q^2)= D(\T).
  \label{DmQ-eq}
\end{equation}
Indeed, $D(Q^2)=D((TT^*)^{-1})=R(T)=D(\T)$ so it remains to be shown that
$D(Q)=D(\T)$. But if $(x_k)$ is a sequence in $D(Q^2)=D(\T)$,
\begin{equation}
  \nrm{Q(x_k-x_m)}{}^2=\scal{Q^2(x_k-x_m)}{x_k-x_m}=\nrm{\T(x_k-x_m)}{}^2,
\end{equation}
whence the closures of this set in the graph topologies on $D(Q)$ and $D(\T)$
coincide.

Combining the boundedness of the resolvent with \eqref{DmQ-eq} it follows
that
\begin{align}
  (I+\T^*\T)^{-1}&\colon F_1\to D(\T^*\T)=D(Q^2)  \\
  Q(I+\T^*\T)^{-1}&\colon F_1\to D(Q)=D(Q^2)  \\
  Q^2(I+\T^*\T)^{-1}&\colon F_1\to D(Q)=D(\T),  
\end{align}
and it is seen from the first of these lines and \eqref{DmQ-eq} that
$\T(I+\T^*\T)^{-1}$ belongs to $\B(F_1,F)$; then the third line gives that
$\T\T^*\T(I+\T^*\T)^{-1}\in\B(F_1,F)$, and \eqref{Tres-eq} implies that
$R(T)=D(\T)=F_1$, which is closed. 

Given that \eqref{RT-cnd}--\eqref{prj-cnd} hold,  then \eqref{S-cnd} gives
both that $S^{\frac12}$ is injective and that $S^{-\frac12}\in\B(F)$,
because it is closed and everywhere defined, and similarly
$U:=TS^{-\frac12}\in\B(F,F_1)$.

For any $x\in D(S)$ it holds that $S^{-\frac12}x\in D(S)\subset D(T^*T)$, so
\begin{equation}
  \nrm{Ux}{}^2=\scal{S^{-\frac12}T^*TS^{-\frac12}x}{x}=\nrm{x}{}^2,
\end{equation}
This extends to
all $x\in F$, and $TS^{-\frac12}$ maps onto $F_1$ by \eqref{RT-cnd}, for the
fact that $D(S^{\frac12})=D(T|_F)$ (shown analogously to \eqref{DmQ-eq})
shows that $S^{-\frac12}$ maps onto $D(T)\ominus Z(T)$. Hence $U$ is unitary.

The identity \eqref{iseq-eq} holds on vectors $x\in D(S^2)$, for
 $S^2x=(T^*T)^2x=T^*S_1Tx$, so that
\begin{equation}
  (S^{-\frac12}T^*)S_1(TS^{-\frac12})x=
  S^{-\frac12}S^2S^{-\frac12}x=I_{D(S^{1/2})}S x= S x;
\end{equation}
note that $S^{-\frac12}T^*\subset U^*$.
Now any $x\in D(S)$ may be approximated in the graph norm
by $x_k\in D(S^2)$: then $Sx_k\to Sx$ while
\begin{equation}
  S_1 Ux_k=US x_k\to USx\quad\text{and}\quad Ux_k\to Ux;
\end{equation}
therefore $Ux\in D(S_1)$ with $S_1Ux=USx$, from where \eqref{iseq-eq}
follows for $x$ by application of $U^*$. Conversely, note that $U^*S_1U$
acts as  $U^*TS^{\frac12}$, hence has its domain contained in 
$D(TS^{\frac12})=D(S)$ (since $U^*\in \B(F_1,F)$); therefore $U^*S_1Ux$ can
only be defined when $Sx$ is so, and then we have already seen that $U^*S_1Ux=Sx$.

Finally, for $\lambda\in\sigma(S)$ there is $x_k\in D(S)$ with
$\nrm{x_k}{}=1$ and $(S-\lambda)x_k\to0$, and $y_k=Ux_k$ is also normalised
while
\begin{equation}
  (S_1-\lambda)y_k=U(U^*S_1U-\lambda)x_k\to0 \quad\text{for}\quad k\to\infty;
\end{equation}
moreover, $y_k\to0$ weakly if the $x_k$ do so, hence also
$\sigma_{\ess}(S)\subset \sigma_{\ess}(S_1)$, and the opposite inclusions
are equally easy.
\end{proof}

The requirement in \eqref{S1-cnd} is equivalent to $0$ belonging to the
resolvent set of $S_1$, and by the minimax principle this may, of course, be
replaced by strict positivity of $S_1$.    
Applied to  the complex \eqref{dmu-cmplx} this yields, because of \eqref{AAinv-eq}:

\begin{cor}
  \label{iseq-cor}
The conclusions of Theorem~\ref{kk1-thm} are valid.
\end{cor}

For $k=0$ this almost gives the main part of
Theorem~\ref{Brascamp--Lieb-thm}, for clearly $\A=d^*_0d_0$ has
kernel $\C$ so that $P_0$ must equal $u-\int u\,d\mu$; whence
\eqref{P-eq}. 

The goal is not yet attained, however. First of all we shall 
 in Section~\ref{Ak-sect} below give a definition of $\A$ and $\AA$ using
sesqui-linear forms, and then verify in
\eqref{Ak-id'} and \eqref{Ak-eq} below that this
coincides with the $A_k$ in \eqref{Ak-def} above and gives a meaning to
\eqref{AAA-eq}.
Secondly, the formula for $P_k$ in Theorem~\ref{kk1-thm} is obtained for
$H^1(\mu,\wedge^k\Cn)$ only, whereas for 
Theorem~\ref{Brascamp--Lieb-thm} it is necessary 
to make sense of the right hand side of \eqref{P-eq} when the $u$ there is
arbitrary in $L^2(\mu)$. This is based on the Lax--Milgram definition in
Section~\ref{Ak-sect}, and is carried out in Section~\ref{d-sect}.

\begin{rem}
  \label{spec-rem}
Corollary~\ref{iseq-cor} and the remark following it entails 
\begin{equation}
  \sigma(\A)=\{0\}\cup \sigma(\AA|_{X_1}).
\end{equation}
Earlier Sj{\"o}strand \cite{Sj96} obtained that the gap between
the first two 
eigenvalues of $\A$ is larger than $\AA$'s first eigenvalue. This also
follows immediately from the above formula when the assumptions on $\mu$, or
$\Phi$, imply that the spectra are discrete, as in \cite{Sj96}. 
\end{rem}


\section{The $d$-complex in Weighted $L^2$ Spaces}   \label{Ak-sect}

From this section and onwards, the probability measure is assumed to have
the form $d\mu=e^{-\Phi(x)}\,dx$ in order to derive more explicit conditions.

\subsection{The Operators}
  \label{min-ssect}
To avoid cumbersome justification of integration by parts, 
it is worthwhile to define $\A$, $\AA$, \dots\  variationally, i.e.\ 
by Lax--Milgram's lemma (in contrast with \cite{Hel93,Sj96,Hel96}
were the Friedrichs extension was used).  This is based on the weighted
space $L^2(\mu):= L^2(\Rn,\mu,\C)$ with measure $\mu:= e^{-\Phi(x)}dx$ and scalar product
\begin{equation}
  \scal{u_1}{u_2}_{\mu}:= \int_{\Rn} u_1(x)\overline{u_2(x)}\,d\mu(x)
\end{equation}
and its analogue for $k$-forms $L^2(\Rn, \mu,\wedge^k\Cn)$, where
$\scal{v}{w}_{\mu}$ is defined instead by integration of $\sum{}'
v_J(x)\overline{w_J(x)}$, with prime denoting summation over increasing $k$-tuples
$J$, see Section~\ref{notation-sect}. 

For each $k$ the operator $\Ak$ is defined in $H_k:=L^2(\Rn,\mu,\wedge^k\Cn)$ 
by means of the triple $(H_k, V_k, a_k)$, where
\begin{align}
  a_k(v,w)&=\scal{d_kv}{d_kw}_{\mu} +\scal{d^*_{k-1}v}{d^*_{k-1}w}_{\mu}
  \qquad (d^*_{-1}\equiv 0)
  \label{ak-eq} \\
  V_k&= \bigl\{\,v\in L^2(\mu,\wedge^k\Cn) \bigm| d_kv\in H_{k+1},\quad
         d^*_{k-1}v \in H_{k-1} \,\bigr\},
  \label{V-eq}
\end{align}
with $a_k(\cdot,\cdot)$ defined on $V_k$ and $d_k$ equal to the
exterior differential from $D(d_k)\subset H_k$ to $H_{k+1}$ (with
derivatives calculated in $\cal D'(\Rn)$), while $d^*_k$ denotes the Hilbert
space adjoint with
respect to $\scal{\cdot}{\cdot}_{\mu}$, see \eqref{d*-id} below for the
expression. Recall that in this way $\Ak$ is defined as follows:
\begin{align}
  D(\Ak)&= \bigl\{\,u\in V_k \bigm| \exists f\in H_k: 
    a_k(u,w)=\scal{f}{w}_{\mu}\ \forall w\in V_k\,\bigr\}
  \label{Alm-eq} 
 \\
  \Ak u&= f \quad\text{for $u\in D(\Ak)$}.
  \label{Alm'-eq}
\end{align}

To substantiate this, note that $V_k$ in \eqref{V-eq}, in view of
$d_k^*$'s closedness and differentiation's continuity in $\cal D'(\Rn)$,
is a Hilbert space with   
\begin{equation}
  \begin{split}
  \scal{v}{w}_{V_k}&:= \scal{v}{w}_{H_k}+\scal{d_kv}{d_kw}_{H_{k+1}}
   +\scal{d^*_kv}{d^*_kw}_{H_{k-1}}
\\
  &\phantom{:}
  =\scal{v}{w}_{\mu}+a_k(v,w);
  \end{split}
  \label{Vnrm-eq}
\end{equation}
the sesqui-linear form is clearly bounded
\begin{equation}
  |a_k(v,w)|\le \norm{v}{V_k}\norm{w}{V_k}\quad \forall v,w\in V_k.
  \label{ak-bd}
\end{equation}
Since $a_k(\cdot,\cdot)$ is symmetric and
\eqref{Vnrm-eq} yields
\begin{equation}
  \Re a_k(v,v)\ge \norm{v}{V_k}^2-\norm{v}{H_k}^2 \ge 0,
\end{equation}
$\Ak$ is well defined, self-adjoint with spectrum in
$[0,\infty[\,$ and with domain $D(\Ak)$ dense in $V_k$. 
When obtaining this from Lemma~\ref{Lax--Milgram-lem}, it is
important to have density of $V_k\subset H_k$, but more than that holds in
the present set-up:

\begin{lem}
  \label{dns-lem} The space $C^\infty_0(\Rn,\wedge^k\Cn)$ is dense in each of
the domains $D(d_k)$, $D(d^*_{k-1})$ and $D(d_k)\cap D(d^*_{k-1})=V_k$ with
respect to their graph topologies (for $V_k$ this is given by \eqref{Vnrm-eq}).
\end{lem}
This may be proved by a cut-off and convolution procedure, as in 
\cite[Lem.~4.1.3]{Hdbar} mutatis mutandem. (Specifically one should let
$\Omega=\Rn$, $S=d_k$ and $T^*=d^*_{k-1}$ there, while $\varphi_1$,
$\varphi_2$, $\varphi_3$ should equal our $\Phi$; the required inequality
(4.1.6) there is redundant for we may take
$\eta_\nu(x)=\eta(\nu^{-1}x)$ on~$\Rn$.)

As a second application of Lemma~\ref{dns-lem} we have a
characterisation of $\Ak$:

\begin{lem}
  \label{fx-lem}
$\Ak$ equals the Friedrichs extension from $C^\infty_0(\Rn,\wedge^k\Cn)$.
\end{lem}
\begin{proof}
Let $S$ denote $\Ak$'s restriction to $C^\infty_0$ and let $T$ be
the Friedrichs extension (using that $\Ak\ge0$). The completion of
$C^\infty_0$ with respect to 
\begin{equation}
  \norm{u}{V_S}^2:= \scal{Su+(1-m(S))u}{u}_{\mu}
\end{equation}
is a Hilbert space $V_S\subset L^2(\mu,\wedge^k\Cn)$. If
$\varphi\in C^\infty_0$ and $c=1+m(S)_-$,
\begin{equation}
  \norm{\varphi}{V_S}^2= a_k(\varphi,\varphi) +
     (1-m(S))\nrm{\varphi}{\mu}^2 \le c \norm{\varphi}{V_k}^2
\end{equation}
so by Lemma~\ref{dns-lem} it is found that $D(\Ak)\subset V_k\subset V_S$.
Because $T$ is the only self-adjoint extension of $S$ with domain
contained in $V_S$, this yields $\Ak=T$.
\end{proof}

As a part of the above-suggested proof of Lemma~\ref{dns-lem} it is found
that $d^*_{k-1}$ has the following action on forms $f=\sum'_{|J|=k}
f_J\,dz^J$, cf.~Section~\ref{notation-sect}, when derivatives are
calculated in $\cal D'$: 
\begin{equation}
  d^*_{k-1}f= \sum_{L}^{}{}'(\sum_{j=1}^n
(\Phi'_j-\partial_j)f_{jL})\,dz^L. 
  \label{d*-id}
\end{equation}
For later reference the argument is recalled: if $f\in H_k$ and $w\in
C^1_0(\Rn,\wedge^{k-1}\Cn)$,
\begin{equation}
  \begin{split}
    \scal{f}{dw}_{\mu}&= \int\sum_{J}^{}{}'f_J \sum_{j=1}^n\sum_{L}^{}{}'
       \partial_j\overline{w_L}\cdot \varepsilon^{jL}_{J} e^{-\Phi}\,dx
\\
   &=\sum_{L}^{}{}'\dual{\sum_{j,J}-\partial_j(e^{-\Phi}f_J)
                            \varepsilon^{jL}_{J}}{\overline{w_L}}
\\
   &=\sum_{L}^{}{}'\dual{\sum_{j}(\Phi'_j-\partial_j)f_{jL}}
                            {e^{-\Phi}\overline{w_L}}.
  \end{split}
  \label{d*-eq}
\end{equation}
This should be justified since $w$ is not $C^{\infty}$, but by reading
$\dual{\cdot}{\cdot}$ as the duality of $\cal D'{}^1$ and $C^1_0$, the
$f_J$ may be approximated from $C^\infty_0(\Rn)$ so that
Leibniz' rule may be applied together with the continuity of
$\partial_j\colon \cal D'{}^0\to\cal D'{}^1$ in the last~line.

Taking $f\in D(d^*_{k-1})$ and $w$ such that $w_L=\delta_{LK}\varphi$ for
arbitrary $K$ and $\varphi\in C^\infty_0(\Rn)$, the left hand side equals
$\dual{(d^*_{k-1}f)_{K}}{e^{-\Phi}\overline{\varphi}}$ so that \eqref{d*-id} results.

Using Lemma~\ref{dns-lem} once more, we now observe (as a first step in the
verification of \eqref{Ak-def}) two different identifications for the domain
of $d^*_{k-1}$:

\begin{lem}
  \label{maxmin-lem}
For the exterior differential $d$ going from
$H_{k-1}$ to $H_k$, whereby $H_j$ stands for the space
$L^2(\Rn,\mu,\wedge^j\Cn)$, and its formal
adjoint $d^*$ given by \eqref{d*-id}, the minimal and maximal realisation coincide, 
i.e.\ $d_{\op{min}}=d_{\op{max}}$ and $(d^*)_{\op{min}}=(d^*)_{\op{max}}$.
\end{lem}
Indeed, for $d$ itself this is just a reformulation of Lemma~\ref{dns-lem},
so by duality $(d^*)_{\op{max}}=(d^*)_{\op{min}}=d^*_{k-1}$. For this
reason the true adjoint $d_k^*$ may be abbreviated as $d^*$.

\bigskip

When it is understood that $d$ and $d^*$ act in the distribution sense (as
opposed to their maximal realisations in $H_k$, viz.\ $d_k$ and $d^*_{k-1}$,
which have subscripts), it
is now easy to infer that $\Ak$ acts according to the formula
\begin{equation}
  \Ak= d^*d+ dd^*.
  \label{Ak-id'}
\end{equation}
For this it is advantageous to test against $w=e^{\Phi}\varphi$ for
$\varphi\in C^\infty_0(\Rn,\wedge^k\Cn)$:
\begin{equation}
  a_k(v,w)=\scal{\Ak v}{w}_{\mu}=\scal{\Ak v}{\varphi}
  \label{act1-eq}
\end{equation} 
if $v\in D(\Ak)$, while taking $f=d_kv\in H_{k+1}$ in \eqref{d*-eq} yields
\begin{equation}
  \scal{d_kv}{d_kw}_{\mu}=\dual{d^*d v}{\overline{\varphi}}.
\end{equation}
Strictly speaking the right hand side should be read as a sum (over $|J|=k$)
of distributions acting on $\varphi_J$, cf.~\eqref{d*-eq}, for the dual of
$C^\infty_0(\Rn,\wedge^k\C^n)$ is not considered here.
Using the compact support of $w$ it follows analogously to
\eqref{d*-eq} that, since $d^*v$ is in $\cal D'{}^0$ (or rather has
coefficients there),
\begin{equation}
  \dual{dd^*v}{\overline{\varphi}}=
  \sum_{J}^{}{}'\dual{\sum_{j,L}\partial_j(d^*v)_L\varepsilon^{jL}_{J}}{
               e^{-\Phi}\overline{w}_J}= 
  \scal{d^*_{k-1}v}{d^*_{k-1}w}_\mu.
\end{equation}
From the definition of $a_k$ this shows \eqref{Ak-id'}.

Combining \eqref{Ak-id'} with \eqref{d*-id} a calculation now yields an
explicit formula for $\Ak$'s action. The details of this will be given
partly to verify the expressions for $\A$ and $\AA$ in the introduction, and
partly because such formulae may be of interest in their own right.

Since $dv=\sum'_K\sum_{m,M}\partial_mv_M\varepsilon^{mM}_K\,dz^K$,
where $|K|=|M|+1$,
\begin{equation}
  \begin{split}
  d^*dv&=\sum_J^{}{}'(\sum_j(\Phi'_j-\partial_j)
              \sum_{K}^{}{}'\varepsilon^{jJ}_{K}(dv)_K)\,dz^J
\\
  &=\sum_J^{}{}'\sum_{j,m,M}(\Phi'_j-\partial_j)\partial_mv_M
              \varepsilon^{mM}_{jJ}\,dz^J,
  \end{split}
  \label{Ak-eq'}
\end{equation}
while the other contribution becomes, with $|L|=|J|-1$,
\begin{equation}
  dd^*v= \sum_J^{}{}'(\sum_{j,l,L}\partial_j(\Phi'_l-\partial_l)
           v_{lL}\varepsilon^{jL}_{J})\,dz^J.
  \label{Ak-eq''}
\end{equation}
Taking $j=m$ in \eqref{Ak-eq'} for fixed $J$, only $j\notin J=M$ gives a
non-trivial term, viz.\ $(\Phi'_j-\partial_j)\partial_jv_J$; and for
$j\neq m$ there are contributions when $j\in M$ and $m\in J$, in which
case deletion of $j$ from $M$ and of $m$ from $J$ gives the same tuple,
say $L$, so that 
\begin{equation}
  \varepsilon^{mM}_{jJ}=\varepsilon^{mM}_{jmL}\cdot\varepsilon^{mL}_J
   =-\varepsilon^{mL}_J\cdot\varepsilon^M_{jL}
\end{equation}
and hence
$\partial_mv_M\varepsilon^{mM}_{jJ}=-\partial_mv_{jL}\varepsilon^{mL}_J$.
Therefore 
\begin{equation}
  d^*dv=\sum_J^{}{}'(\sum_{j\notin J}(\Phi'_j-\partial_j)\partial_jv_J+
        \sum_{j\neq m}\sum_{L}^{}{}'\varepsilon^{mL}_{J}
         (\Phi'_j-\partial_j)\partial_m v_{jL})\,dz^J.
  \label{Ak-eq'''}
\end{equation}
When $j=l$ in \eqref{Ak-eq''}, then $\varepsilon^{jL}_J=0$ unless $j$ {\em
belongs\/} to $J$, so $\sum_{L}'$ has only one term $\neq 0$; hence
$\underset{J\ni j}{\sum'\sum}(\Phi'_j-\partial_j)\partial_jv_J\,dz^J$ is the
contribution. For $j\neq l$ there appears a term present in
\eqref{Ak-eq'''} with opposite sign, plus one involving $\Phi''$.

Therefore, when $v=\sum'_{J}v_J\,dz^J$ the action of $d^*d+dd^*$, and
hence of $\Ak$, is altogether given by the relatively simple formula
\begin{equation}
  (d^*d+dd^*)v= \sum_J^{}{}'((\mlap+\Phi'\cdot\nabla)v_J+
     \sum_{j\in J}\sum_{l=1}^n \Phi''_{jl}v_{J_{j\to l}})\,dz^J,
  \label{Ak-eq}
\end{equation}
where $J_{j\to l}$ means $J$ with $j$ replaced by $l$. Note that the degree
$k$ of $v$ really only enters in the determination of how $\Phi''$ acts on $v$.
(This formula was also given by Sj{\"o}strand \cite{Sj96} for the Witten
Laplacians ensuing after the transformation in Section~\ref{unit-ssect} below.)

In particular, if $1$-forms are identified with vector functions,
\begin{equation}
  \AA v= (\mlap+\Phi'\cdot\nabla)\otimes I v+\Phi''\cdot v
\end{equation}
as claimed in the introduction. Note that for $k=0$ or $n$ the action of
$\Ak$ is given by \eqref{Ak-eq'} or \eqref{Ak-eq''}, respectively, that is
\begin{align}
  \A u&= (\mlap+\Phi'\cdot\nabla)u  \quad\text{for}\quad u\in D(\A)
\\
  A_n f&= (\mlap+\Phi'\cdot\nabla)f + (\lap\Phi)f,  
\end{align}
when $f$ in $L^2(\Rn,\mu,\wedge^n\C)$ is (considered as) a function in
$D(A_n)$. 

\subsection{Unitary Transformation}
  \label{unit-ssect}
The operators $\Ak$ are easily transformed into the Witten-Laplacians
denoted by $\lap^{(k)}_{\Phi}$ in \cite{Sj96}. E.g.\
multiplication by $e^{-\Phi/2}$ defines a unitary 
$U\colon L^2(\mu)\to L^2(\Rn)$ transforming $\A$ into (a realisation of)
$B_0=\mlap+\frac14|\Phi'|^2-\frac12\lap\Phi$.

For later reference, this is done consisely here. Writing $d=\sum
\partial_j\,dz_j\wedge$ for the differential on $k$-forms $v=
\sum'_{J} v_J\,dz^J$, where $v_J$ is in $\cal D'(\Rn)$ in general, the formal
adjoint with respect to $\scal{\cdot}{\cdot}_{\mu}$ on
$L^2(\Rn,\mu,\wedge^k\Cn)$ is 
\begin{equation}
  d^*=\sum_{j=1}^n (\Phi'_j-\partial_j)\,dz_j\rfloor
\end{equation}
whereby $dz_j\rfloor$ either removes $dz_j$ when present (and
anti-commuted to the left) or gives zero. When denoting (with subscript
$k$ if necessary)
\begin{gather}
  Uv= \sum_J^{}{}'  e^{-\Phi/2}v_J\,dz^J
  \label{U-def}
\\
  d_{\Phi}= \sum_{j=1}^n (\partial_j+\tfrac{1}{2}\Phi'_j)\,dz_j\wedge,\qquad
  d_{\Phi}^*=\sum_{j=1}^n (-\partial_j+\tfrac{1}{2}\Phi'_j)\,dz_j\rfloor,
\end{gather}
then $d_{\Phi}^*$ is the formal adjoint of $d_{\Phi}$ on
$L^2(\Rn,\wedge^k\Cn)$ and (on $\cal D'{}^1$ forms)
\begin{equation}
  U_{k+1}d_{k}=d_{k,\Phi}U_k,\qquad
  U_kd_{k}^*=d_{k,\Phi}^*U_{k+1}.
  \label{ddfi-id}
\end{equation}
Using this, $v\in D(\Ak)$ with $v_1=\Ak v$ if
and only if for all $w\in V_k$,
\begin{equation}
  \scal{Uv_1}{Uw}=
  \scal{d_{\Phi}Uv}{d_{\Phi}Uw}+
  \scal{d_{\Phi}^*Uv}{d_{\Phi}^*Uw}.
\end{equation}
Hence $U\Ak U^*$ equals the operator $B_k$ for the triple
$(L^2(\Rn,\wedge^k\Cn),V_{k,\Phi}, b_{k})$ where
$b_k(\cdot,\cdot)$ equals
$\scal{d_{\Phi}\cdot}{d_{\Phi}\cdot}+\scal{d_{\Phi}^*\cdot}{d_{\Phi}^*\cdot}$
while $V_{k,\Phi}:=UV_k$ equals $D(d_{k,\Phi})\cap D(d_{k-1,\Phi}^*)$ as
subspaces of $L^2(\Rn,\wedge^k\Cn)$; cf.~Lemma~\ref{Lax--Milgram-lem}.

Using \eqref{ddfi-id} it follows that the $B_k$ acts as $\lap^{(k)}_{\Phi}$.

\subsection{Identification with the Hodge Laplacian}
Denoting by $X_{k+1}$ the closure of $d_k$'s range in
$L^2(\Rn,\mu,\wedge^{k+1}\Cn)$, that is $X_{k+1}=\overline{R(d_k)}$ and
$X_0=\C$ as in \eqref{X-id},
\begin{equation}
  L^2(\Rn,\mu,\wedge^{k+1}\Cn)=H_{k+1}=X_{k+1}\oplus Z(d^*_{k}).
  \label{XZ-id}
\end{equation}
It is now elementary to see that $\Ak$ commutes with the orthogonal
projections onto the summands, and exploiting Lemma~\ref{dns-lem} once more
it also follows that \eqref{Ak-id'} may be read with $d$ and $d^*$ as the
respective unbounded operators associated with the complex \eqref{dmu-cmplx}:

\begin{lem}
  \label{X-lem}
If $P_k$ is the orthogonal projection onto $X_k$,  
\begin{equation}
  P_kV_k\subset V_k, \qquad P_k\Ak\subset \Ak P_k.
  \label{invar-eq}
\end{equation}
Furthermore, the restriction $\Ak|_{X_k}$ is injective, and
$A_k=d^*_kd_k+d_{k-1}d^*_{k-1}$ holds as a formula for unbounded operators, 
i.e.\ with $D(\Ak)=D(d^*_kd_k)\cap D(d_{k-1}d_{k-1}^*)$.
\end{lem}

\begin{proof} Omitting some $k$'s for simplicity,
it follows from \eqref{XZ-id} that $PV\subset V$, since
$V=D(d^*)\cap D(d)$. Using this one finds: if $v\in D(\Ak)$ and $w\in
V$, then $d^*(1-P)\equiv0$ and $dP\equiv0$ so that
\begin{equation*}
  a_1(Pv,w)=\scal{d^*Pv}{d^*w}_{\mu}=\scal{d^*v}{d^*Pw}_{\mu}
       =\scal{\Ak v}{Pw}_{\mu}=\scal{P\Ak v}{w}_{\mu};
\end{equation*}
consequently $P\Ak\subset \Ak P$. 
If $v\in X\cap Z(\Ak)$, then $0=\scal{\Ak v}{v}_{\mu}=\nrm{d^*v}{\mu}^2$,
so $v$ is also in $Z(d^*_{k-1})=X^\perp$, whence $v=0$.

If $v\in D(\Ak)$ and $\Ak v=v_1$ then $\scal{Pv_1}{w}_{\mu}=
\scal{d^*Pv}{d^*w}_{\mu}$ and
$\scal{(1-P)v_1}{w}=\scal{d(1-P)v}{dw}_{\mu}$ 
for all $w$ in $V$. Because $C^\infty_0$ is dense with respect to the graph
norms in $D(d)$ and $D(d^*)$, this gives by closure that
$d^*Pv\in D(d)$ with $dd^*Pv=Pv_1$ and that $d(1-P)v\in D(d^*)$ with
$d^*d(1-P)v=(1-P)v_1$; hence that $v\in D(d^*_1d_1)\cap D(d_0d_0^*)$ with
$dd^* v=Pv_1$ and $d^*dv=(1-P)v_1$. Conversely
$d^*_kd_k+d_{k-1}d_{k-1}^* \subset \Ak$ follows easily from
\eqref{ak-eq}--\eqref{V-eq}.  
\end{proof}

Note that by \eqref{invar-eq} the subspaces $X_k$ and $Z(d^*_{k-1})$ are
invariant under $\Ak$, and that the terms $d^*_kd_k$ and $d_{k-1}d^*_{k-1}$
vanish there, respectively.

Also both $d_{k}\Ak$ and $A_{k+1}d_{k}$ are defined on $D(d_{k+1}d^*_kd_k)$,
so in this way we have the intertwining properties
\begin{equation}
  A_{k+1}d_{k}=d_k\Ak,\qquad   A_{k-1}d^*_{k-1}=d^*_{k-1}\Ak
\end{equation}
on $D(d_kd^*_kd_k)$ and $D(d^*_{k-1}d_{k-1}d^*_{k-1})$, respectively,
for the unbounded operators, as well as in general in the distribution sense.

\bigskip

Since Lemma~\ref{X-lem} shows that the $\Ak$ of this section coincide
with \eqref{Ak-def} above, it is clear that Theorem~\ref{kk1-thm} holds
for the operators given in \eqref{ak-eq}--\eqref{Alm'-eq} and
\eqref{Ak-eq}.

\subsection{A direct $\boldsymbol{H^1}$-proof} 
The injectiveness of $\AA|_X$ shown in Lemma~\ref{X-lem} may be used for
a short proof of Theorem~\ref{Brascamp--Lieb-thm}'s essential
parts. This is done in the spirit of \cite{Hel96,Sj96}, but now for
our general $\Phi$ and with much sharper assumptions:

\begin{prop}
  \label{Brascamp--Lieb-prop}
Suppose \eqref{phi-cnd} holds and that $\A$ defined above 
satisfies:
\begin{equation}
  R(\A)=\overline{R(\A)}\quad\text{in}\quad H_0.
  \label{A-cnd}
\end{equation}
Then it holds true for all $g_1$, $g_2\in H^1(\mu)$ that
\begin{equation}
  \scal{g_1-\ang{g_1}}{g_2-\ang{g_2}}_{\mu} 
                = \scal{\AA^{-1}d g_1}{d g_2}_{\mu}.
  \label{Brascamp--Lieb-eq}
\end{equation}
\end{prop}

\begin{proof}
Since $\A=\A^*$, with closed range by \eqref{A-cnd} and $Z(A_0)=\C$ by
\eqref{Alm-eq}--\eqref{Alm'-eq}, there is a decomposition $L^2(\mu)=
R(\A)\oplus \C$; hence $g_1-\ang{g_1}=\A f$ for some $f\in D(\A)$, so
according to \eqref{Alm-eq},
\begin{equation}
   \scal{g_1-\ang{g_1}}{g_2-\ang{g_2}}_{\mu} 
   =a_0(f,g_2-\ang{g_2}) 
   =\scal{d f}{d g_2}_{\mu}.
  \label{cov-eq}
\end{equation}
In the distribution sense $d^*d f=\A f$, since $f$ is picked in $D(\A)$;
therefore we moreover have for $w\in C^\infty_0(\Rn,\wedge^1\Cn)$, 
\begin{equation}
  a_1(d f,w)=0+ \scal{\A f}{-\dv(e^{-\Phi}w)} 
  =\lim_{k\to\infty}\dual{d\A f}{\overline{\varphi_k}}=\scal{d g_1}{w}_{\mu}
  \label{f-eq}
\end{equation}
when $\varphi_k\in C^\infty_0$ tends to $e^{-\Phi}w$ in the $V_1$-topology.

By completion \eqref{f-eq} also holds for every $w\in V_1$, cf.~the
density in Lemma~\ref{dns-lem} and \eqref{ak-bd},  so it follows that
$d f\in D(\AA)$ with $\AA^{-1}d g_1= d f$ (using the injectivity of
Lemma~\ref{X-lem}). This and \eqref{cov-eq} yields the proof.
\end{proof}

In addition to the above, we may observe that using
\eqref{Brascamp--Lieb-eq}, partial integration gives for each 
test function $\varphi$
\begin{equation}
  \scal{Pu}{\varphi}_{\mu}= \scal{Pu}{P\varphi}_{\mu}
  = \scal{d^*\AA^{-1}d u}{\varphi}_{\mu},
  \label{P'-eq}
\end{equation}
which shows \eqref{P-eq} when $u\in H^1(\mu)$. Since $P$ is bounded in
$L^2(\mu)$, we can extend $d^*\AA^{-1}d$ by
continuity such that \eqref{P-eq} holds by definition on $L^2(\mu)$.

A closer analysis given in Section~\ref{prf-ssect} below will show
that each of the individual factors in $d^*\AA^{-1}du$ are well defined too, when
$u\in L^2(\mu)$.

\subsection{Brascamp--Lieb's inequality}
  \label{Brascamp--Lieb-ssect}
When $\Phi$ is strictly convex, Corollary~\ref{Brascamp--Lieb-cor} may now
be proved for $f\in H^{1}(\mu)$, i.e.\
\begin{equation}
  \norm{f-\ang{f}}{L^2(\mu)}^2 \le \scal{(\Phi'')^{-1}d f}{d f}_{\mu}.
  \label{Brascamp--Lieb-eq'}
\end{equation}
To begin with it is first assumed that $\Phi''(x)\ge c_0>0$ in $\Cn$ for
each $x\in\Rn$. Partial integration shows that on $C^\infty_0$ we have
\begin{equation}
  \scal{\AA v}{v}_{\mu}= a_1(v,v)=\sum_{j,k=1}^n
    \nrm{\partial_jv_k}{\mu}^2 +\scal{\Phi'' v}{v}_{\mu}
  \ge c_0\nrm{v}{\mu}^2
  \label{Brascamp--Lieb-lb}
\end{equation}
so $m(\AA)\ge c_0>0$ follows by Lemma~\ref{fx-lem}. Then
Theorem~\ref{kk1-thm} yields that $\A$ restricted to $\C^\perp$ has $0$ in
the resolvent, hence that \eqref{A-cnd} and \eqref{Brascamp--Lieb-eq} hold.

Because $\AA^{-1}$ is symmetric $\ge0$, Cauchy--Schwarz' inequality
(for such operators) applied to
$\scal{\AA^{-1}v}{\AA w}_{\mu}=\scal{v}{w}_{\mu}$ 
for $v\in L^2(\mu,\wedge^1\C)$ and $w\in D(\AA)$ yields
\begin{equation}
  |\scal{v}{w}_{\mu}|^2\le \scal{\AA^{-1}v}{v}_{\mu}
\scal{w}{\AA w}_{\mu}
\end{equation}
with equality if $v=\AA w$. Hence
\begin{equation}
  \scal{\AA^{-1}v}{v}_{\mu}=\sup
  \bigl\{\,\frac{|\scal{v}{w}_{\mu}|^2}{
   \scal{\AA w}{w}_{\mu}} \bigm| w\in D(\AA)\setminus\{0\}\,\bigr\},
\end{equation}
and analogously for $\Phi''$, so \eqref{Brascamp--Lieb-lb}
and the density of $D(\AA)$ and $C^\infty_0$ in $V$ give
\begin{equation}
  \scal{\AA^{-1}v}{v}_{\mu} 
  = \sup\bigl\{\,\frac{|\scal{v}{w}_{\mu}|^2}{
    a_1(w,w)} \bigm| w\in C^\infty_0(\Rn,\wedge^1\Cn)\setminus\{0\}\,\bigr\}
  \le\scal{(\Phi'')^{-1}v}{v}_{\mu}
  \label{APhi-eq}
\end{equation}
[regardless of whether $C^\infty_0$ is dense in $D(\Phi'')$],
which proves \eqref{Brascamp--Lieb-eq'} in this case.

In general this applies for $0<\varepsilon<1$ to
\begin{equation}
  \Phi_\varepsilon(x)=\Phi(x)+\varepsilon |x|^2 +\log C_\varepsilon
  \quad\text{with}\quad C_\varepsilon=\int e^{-\Phi(x)-\varepsilon |x|^2}
  \,dx,
\end{equation}
which is uniformly strictly convex with $\int d\mu_\varepsilon=1$; note that
$C_\varepsilon\nearrow 1$ for $\varepsilon\searrow 0$. Clearly
\eqref{Brascamp--Lieb-eq'} holds with $\mu_\varepsilon$ instead of $\mu$;
because $e^{-\varepsilon |x|^2-\log C_\varepsilon}\le C_1^{-1}$,
\begin{equation}
  \int |f- {\textstyle\int}f e^{-\Phi_\varepsilon}\,dx|^2
e^{-\Phi_\varepsilon}\,dx \longrightarrow \nrm{f-\ang{f}}{\mu}^2
\end{equation}
by majorised convergence for $\varepsilon\searrow 0$. Indeed, in this way
$\int f e^{-\Phi_\varepsilon}\,dx$ tends to $\ang{f}$ and the whole left hand
side is controlled by $(|f(x)|+\norm{f}{L^1(\mu)})^2 e^{-\Phi(x)}$.

Being positive, $I_\varepsilon := \nabla f^T(\Phi''_\varepsilon)^{-1}
\overline{\nabla f}$
always has an integral; if this is finite for
$\varepsilon=0$ then \eqref{Brascamp--Lieb-eq'} must be proved. But then
$I_0$ itself may serve as a majorant, and because $\Phi''_\varepsilon(x)\ge
\Phi''(x)$ in $\B (\ell^2(\{1,\dots,n\}))$ for all $x$,
\begin{equation}
  I_\varepsilon(x) e^{-\varepsilon |x|^2-\log C_\varepsilon}\le I_0(x)/C_1.
\end{equation}
Pointwise convergence is clear from the norm continuity of
inversion. This completes the proof for $f\in H^1(\mu)$.

\section{Extensions to Integrable Functions}   \label{d-sect}

\subsection{Sufficient Conditions Revisited}
Instead of merely establishing estimates sufficient for the full proof
of Theorem~\ref{Brascamp--Lieb-thm}, their relation to the other
conditions treated is given below, for they all fit so well together
that a discussion should be of interest in its own right; the following ten
conditions are elucidated (in and) 
after the proof. Note that subscripts are suppressed on $d$ and $d^*$ when
the context makes it clear what the domain is.

\begin{thm}
  \label{ran-thm}
For $\Ak=d^*_kd_k+d_{k-1}d^*_{k-1}$ with $k>0$ it holds true that
\begin{equation}
  \begin{gathered}
  \text{\eqref{m1-cnd}, \eqref{eqvi-cnd} and \eqref{V-cnd} are
        equivalent }  \\
  \text{\eqref{V-cnd}$\implies $\eqref{Zd-cnd}$\implies $\eqref{X-cnd}}
\\
  \text{\eqref{X-cnd}, \eqref{d*-cnd}, \eqref{Rd-cnd}, \eqref{R0-cnd}, \eqref{eqviX-cnd}
   and \eqref{m1X-cnd} are equivalent}, 
  \end{gathered}
  \label{impl-eq}
\end{equation} 
when the properties \eqref{m1-cnd}--\eqref{m1X-cnd} are as follows:
\begin{enumerate}
\addtolength{\itemsep}{\jot}
\addtolength{\topsep}{3\jot}
  \item  \label{m1-cnd} 
$0< m(\Ak):=\inf\{\,\scal{\Ak v}{v}_{\mu}\mid v\in D(\Ak),\
\nrm{v}{\mu}=1,\}$;
  \item \label{eqvi-cnd} $\text{the norms $a_k(\cdot,\cdot)^{1/2}$ and
$\norm{\cdot}{V_k}$ are equivalent on $V_k$}$;
  \item  \label{V-cnd}
$\nrm{v}{\mu}^2\le c^2(\nrm{d^*v}{\mu}^2+\nrm{dv}{\mu}^2) 
\text{ for all $k$-forms } v\in D(d^*)\cap D(d)$;
  \item  \label{Zd-cnd}
$\nrm{v}{\mu}^2\le c^2(\nrm{d^*v}{\mu}^2+\nrm{dv}{\mu}^2) 
\text{ for all $k$-forms } v\in D(d^*)\cap Z(d)$;
  \item  \label{X-cnd}
$\nrm{v}{\mu}^2\le c^2(\nrm{d^*v}{\mu}^2+\nrm{dv}{\mu}^2) 
\text{ for all $k$-forms } v\in D(d^*)\cap X_k$;
  \item  \label{d*-cnd}
$\nrm{v}{\mu}\le c\nrm{d^*v}{\mu} 
\text{ for all $k$-forms } v\in D(d^*)\cap X_k$;
  \item  \label{Rd-cnd}
the range of 
$d\colon D(d_{k-1})\to L^2(\mu,\wedge^{k}\C) \text{ is closed, i.e.~}
R(d_{k-1})=X_{k}$; 
  \item  \label{R0-cnd} 
the range $R(A_{k-1}|_{Z(d)^\perp})$ is closed in
$L^2(\mu,\wedge^{k-1})$; 
  \item  \label{eqviX-cnd}
$\text{the norms $\nrm{d^*\cdot}{\mu}$ and
$\norm{\cdot}{V_k}$ are equivalent on $V_k\cap X_k$}$;
  \item  \label{m1X-cnd}
$0< m(\Ak|_{X_k}):=\inf\{\,\scal{\Ak v}{v}_{\mu}\mid v\in D(\Ak)\cap X_k,\
\nrm{v}{\mu}=1,\}$.
\end{enumerate}
In the affirmative case $m(\Ak)=m(\Ak|_{X_k})\ge c^{-2}$, 
where $c$ is any of the constants in \eqref{V-cnd}--\eqref{d*-cnd}.

Moreover, the closed forms in $L^2(\mu,\wedge^k\Cn)$ belong to $R(d_{k-1})$,
i.e.~$Z(d_k)=R(d_{k-1})$, when any of \eqref{m1-cnd}--\eqref{Zd-cnd} holds.
\end{thm}

\begin{proof}
\eqref{m1-cnd}$\iff$\eqref{eqvi-cnd}  because they are both equivalent to
\eqref{V-cnd} in view of $D(\Ak)$'s density in $V_k$ and
\eqref{ak-eq}--\eqref{Vnrm-eq}. In addition $m(\AA)\ge c^{-2}$ is found.

Now \eqref{V-cnd}$\implies $\eqref{Zd-cnd}$\implies
$\eqref{X-cnd}$\iff$\eqref{d*-cnd} by shrinking of the set $F_1$ of
$v$'s from $D(d)$ to $X_k$ (and since $d\equiv0$ on $X$). That
\eqref{d*-cnd}$\iff$\eqref{Rd-cnd}$\iff$\eqref{R0-cnd} follows from
Theorem~\ref{CR-thm} with $F_1=X_k$; this moreover gives that
$R(d_{k-1})=Z(d_k)$ when \eqref{Zd-cnd} holds.

Finally \eqref{d*-cnd} trivially gives
\eqref{eqviX-cnd}, and \eqref{eqviX-cnd}$\implies$\eqref{m1X-cnd} is
clear. When \eqref{m1X-cnd} holds, the inequality in \eqref{d*-cnd} is valid
in the subspace $D(\Ak)\cap X_k$. Therefore, if $v_0\in V_k\cap X_k$ then 
$v_m\underset{m}{\rightarrow}v_0$ in $V_k$ for some $v_m\in D(\Ak)$ where $P_kv_m\to v_0$ in
$L^2(\mu,\wedge^k\Cn)$ with $P_kv_m\in D(\Ak)\cap X_k$ by
Lemma~\ref{X-lem}. Using \eqref{m1X-cnd} and $d^*(1-P)\equiv0$,  
\begin{equation}
  \nrm{v_0}{\mu}\le c\lim_m\nrm{d^*Pv_m}{\mu}=c\lim_m\nrm{d^*v_m}{\mu}
  =c\nrm{d^*v_0}{\mu}.
\end{equation}
Consequently \eqref{d*-cnd} holds.
\end{proof}
While \eqref{R0-cnd} is the central point because of its implications
for Theorem~\ref{Brascamp--Lieb-thm}, \eqref{m1-cnd} and \eqref{m1X-cnd} are 
probably the most convenient to verify for a given $\Phi(x)$; but if this
can be done, then $\Ak^{-1}$ and $\Ak|_{X_{k}}^{-1}$ extend automatically,
by the Poincar\'e inequalities \eqref{eqvi-cnd}, \eqref{eqviX-cnd} and
$2^{\circ}$ of Lemma~\ref{Lax--Milgram-lem}, to 
bounded operators from certain spaces of distributions of order $1$ (or of
forms with such coefficients).

Since \eqref{V-cnd}$\implies$\eqref{Zd-cnd}$\implies$\eqref{X-cnd}, it is
clear that \eqref{Zd-cnd} is an intermediate property in comparison with the
two situations described by \eqref{m1-cnd}--\eqref{V-cnd} on the one hand and
\eqref{X-cnd}--\eqref{m1X-cnd} on the other. 

Actually condition \eqref{Zd-cnd} is equivalent to exactness of the
$d$-complex at $L^2(\mu,\wedge^k\Cn)$: $Z(d_k)$ is closed in
$L^2(\mu,\wedge^k\Cn)$ because $d$ has the maximal domain there and is
continuous in $\cal D'$, so Theorem~\ref{CR-thm} with $F_1=Z(d_k)$ shows that
\eqref{Zd-cnd} holds if and only if $R(d_{k-1})=F_1$.

Injectiveness of $\Ak$ is furthermore a consequence of \eqref{Zd-cnd}. For by
the Lax--Milgram definition of $\Ak$, cf.~\eqref{ak-eq},
\begin{equation}
  Z(\Ak)= Z(d^*_{k-1})\cap Z(d_k),
\end{equation}
so \eqref{Zd-cnd} implies $Z(\Ak)=\{0\}$. In addition, if \eqref{X-cnd} holds
but \eqref{Zd-cnd} doesn't, then there is some $v\in Z(d_k)\setminus
R(d_{k-1})$; writing this as $v=x+x^\perp$ with $x\in X_k$ and 
$x^\perp\in X_k^\perp$, then $x\in R(d_{k-1})$ since \eqref{X-cnd}
implies \eqref{Rd-cnd}, whence $0\ne x^\perp\in Z(d_k)\cap
Z(d^*_{k-1})$. So when $\Ak$ is injective, then either \eqref{Zd-cnd},
i.e.~exactness, holds or $R(A_{k-1}|_{Z(d)^\perp})$ is unclosed.

\subsection{Proof preparations}
As mentioned, \eqref{eqvi-cnd} and \eqref{eqviX-cnd}  
imply the extendability of $\AA^{-1}$ and $\AA|_{X}^{-1}$ to larger spaces than
just the $L^2$-forms, which is crucial for
Theorem~\ref{Brascamp--Lieb-thm}: 

\begin{cor}
  \label{V'-cor}
$1^\circ$ Let \eqref{eqvi-cnd} in Theorem~\ref{ran-thm} hold. 
One has then $\Ak\subset\cal A_k$, when $\cal A_k\colon
V_k\overset{\sim}{\longrightarrow} V'_k$ is the linear isometry  from
$(V_k,a_k(\cdot,\cdot))$ onto its dual $V'$ given by
\begin{equation}
 \dual{\cal A_kv}{\cdot}_{V'_k\times V_k}=a_k(\cdot,\overline{v})
  \quad\text{for}\quad v\in V_k;
  \label{V'-eq}
\end{equation}
and $D(\Ak)=\cal A_k^{-1}(H_k)$ holds.

$2^\circ$
When \eqref{eqviX-cnd} holds, then $\Ak|_{X_k} \subset \tilde{\cal A}_k$, when 
$\tilde{\cal A}_k$ is the isomorphism
$\tilde V_k\overset{\sim}{\longrightarrow} (\tilde V_k)'$,
with $\tilde V_k$ denoting $V_k\cap X_k$ normed by
$\scal{d^*v}{d^*v}_{\mu}^{1/2}$. 
\end{cor}

Observe that when also $H_k\simeq H'_k$ in \emph{this} manner (i.e.\
$g\mapsto \scal{\cdot}{\overline{g}}_{\mu}$ instead of $\scal{\cdot}{g}$), then
there are inclusions 
$V_k\overset{\iota}{\hookrightarrow} H_k\overset{\iota'}{\hookrightarrow}V'_k$ 
with dense ranges and $\iota'$ equal to the transpose of~$\iota$:
\begin{equation}
  \dual{\iota'f}{v}_{V'_k\times V_k}=
  \scal{\iota v}{\overline{f}}_{\mu}=\scal{f}{\overline{\iota v}}_{\mu} 
  \quad\text{for all $f\in H_k$, $v\in V_k$}.
  \label{VHV'-eq}
\end{equation}
Similarly $\tilde V_k \overset{\iota}{\hookrightarrow}
X_k \overset{\iota'}{\hookrightarrow} (\tilde V_k)'$. Thus it is
meaningful to state the last part of $1^\circ$ or the corresponding fact
that $D(\Ak|_{X_k})=\tilde{\cal A}_k^{-1}(X_k)$. 

\begin{proof}
$a_k(u,\overline{w})=\scal{\Ak u}{\overline{w}}_\mu$ $\forall w\in V_k$, 
so \eqref{VHV'-eq} gives $1^\circ$.
For $2^\circ$, identify $\Ak|_{X_k}$ with the operator $\tilde\Ak$
defined from $(X_k,\tilde V_k, \scal{d^*\cdot}{d^*\cdot}_{\mu})$;
completeness of $\tilde V_k$ and denseness carry over from $(H_k,V_k,a_k)$
by means of $P_k$ in Lemma~\ref{X-lem}.  
While $\Ak|_{X_k}\subset \tilde{\Ak}$ is clear, any $v\in D(\tilde\Ak)$
gives $\scal{\tilde\Ak v}{w}_{\mu}=\scal{d^*v}{d^*w}_{\mu}$ for all $w$ in
$V_k$ since $\tilde\Ak v\perp (1-P_k)w$ and
$d^*(1-P_k)w=0$. Since $v\in X_k$ we first get $\scal{\tilde\Ak
v}{w}_{\mu}=a_k(v,w)$, thereafter $\tilde\Ak\subset\Ak$. 
\end{proof}

When applying this we shall need that $V'_k$ or $\tilde V'_k$ can receive the
image $d(H_{k-1})$.
To establish this it is necessary to make a precise identification of $H_k$,
$V'_k$ and $\tilde V'_k$ with subspaces of $\cal D'(\Rn,\wedge^k\Cn)$, which
roughly speaking consists of forms $\sum{}'u_J\,dz^J$ with
$u_J\in\cal D'(\Rn)$. 

In Appendix~\ref{curr-app} below this is introduced concisely by means of a
direct approach based on the finite dimension of $\wedge^k\Cn$ and the
simplicity of the manifold $\Rn$.
This should provide the reader with an alternative to the general and
vast expositions of G.~de~Rham~\cite{Rham55} and L.~Schwartz \cite{Swz59}
and \cite[Ch.~9]{Swz66}. 

By the continuity of $\cal J$ in \eqref{J-eq}, the linear form
$\Lambda_f(\varphi):= \sum{}'\int f_J\varphi_J\,dx$,  
defined on $\varphi\in C^\infty_0(\Rn,\wedge^k\Cn)$, gives injections 
$L^2(\mu,\wedge^k\Cn)\subset L^2_{\loc}(\wedge^k\Cn)
\overset{\Lambda}{\hookrightarrow}\cal D'(\Rn,\wedge^k\Cn)$ for which 
\begin{equation}
  \dual{\Lambda_f}{\varphi}_{\cal D'\times C^\infty_0}
  = \int_{\Rn} f\varphi\,dx=\scal{f}{\overline{e^{\Phi}\varphi}}_{\mu}=
  \dual{\iota' f}{e^{\Phi}\varphi}_{V'_k\times V_k},
  \label{Hembd-eq}
\end{equation}
so $\Lambda$ is the transpose of $M_{e^{-\Phi}}\colon
C^\infty_0(\Rn,\wedge^k\Cn)\to H_k$, multiplication by $e^{-\Phi}$. 

From the last identity above $\Lambda\colon H_k\to\cal D'$ is seen to be
continuous in the topology induced by $V'_k$, hence it extends to $V'_k$ by the
density of $H_k\subset V'_k$:
\begin{prop}
  \label{V'D'-prop}
The operator $\Lambda$ introduced above \eqref{Hembd-eq} extends by continuity to an embedding
$\Lambda\colon V'_k\hookrightarrow \cal D'(\Rn,\wedge^k\Cn)$ and for this it
holds that
\begin{equation}
  \dual{v'}{e^{\Phi}\varphi}_{V'_k\times V_k}=
  \dual{v'}{\varphi}_{\cal D'\times C^\infty_0}.
  \label{Vdual-eq}
\end{equation}
for every $v'\in V'_k$ and $\varphi\in C^\infty_0(\Rn,\wedge^k\Cn)$ 
(when $\Lambda$ is suppressed).
\end{prop}
\begin{proof}
By taking closures, \eqref{Vdual-eq} clearly follows from the left- and rightmost parts
of \eqref{Hembd-eq}. Because $M_{e^{\Phi}}$ has dense range in $V'_k$, the
extended $\Lambda$ is an injection.
\end{proof}
The point of this proposition and \eqref{Vdual-eq} is of course to note the factor $e^{\Phi}$.

From the boundedness of $d^*\colon V_k\to H_{k-1}$ follows the existence of a
bounded transpose $d^{*\prime}\colon H_{k-1}\to V'_k$, and by means of
\eqref{Vdual-eq} and \eqref{d*-id} this is seen to be a realisation of the
distributional differential $d$: indeed for $f\in H_{k-1}$ and elements of
the dense subset $C^2_0\subset V_k$ of the form $e^{\Phi}\varphi$ with
$\varphi$ in $C^\infty_0$,
\begin{equation}
  \begin{split}
  \dual{d^{*\prime}f}{e^{\Phi}\varphi}_{V'_k\times V_k}
  &= \int \sum_{J}^{}{}'f_J \sum_{j=1}^n
  (\Phi'_j-\partial_j)(e^{\Phi}\varphi_{jJ}) e^{-\Phi}\,dx
\\
  &= \sum_{j,J} \dual{\partial_j f_J}{\varphi_{jJ}}
   = \dual{d_{k-1}f}{\varphi},  
  \end{split}
\end{equation}
where the last identity uses \eqref{JJ'-id}, \eqref{dd-id}; by
\eqref{Vdual-eq} this means that $d^{*\prime}f=df$.

The space $\tilde V'_k$ is normed by $\norm{\cdot}{V_k}$ and moreover 
a closed subspace of $V'_k$; indeed $\tilde V'_k=P'_k V'_k$
because $P_k\in\B(V_k)$ (the last fact follows from
$d^*(1-P_k)\equiv0$).  

For one thing this gives an embedding 
$\tilde V'_k\hookrightarrow\cal D'$ by the above construction for $V'_k$,
and for another that $d(H_{k-1})\subset \tilde V'_k$. Indeed, $d^*_{k-1}$
restricts to a continuous map $\tilde V_k\to H_{k-1}$, so $d_{k-1}$ extends
to a bounded operator $H_{k-1}\to \tilde V'_k$; however, this is not a
proper extension since $\tilde V'_k\subset V'_k$.
Altogether we have:

\begin{prop}
  \label{dV'-prop}
The distributional differential is bounded $d\colon H_{k-1}\to V'_k$; it is
the transpose of $d^*\colon V_k\to H_{k-1}$, so for $f\in H_{k-1}$ and
$\varphi\in C^\infty_0(\Rn,\wedge^k\Cn)$ 
\begin{equation}
  \dual{df}{\varphi}_{\cal D'\times C^\infty_0}=
  \dual{df}{e^{\Phi}\varphi}_{V'_k\times V_k}=
  \scal{f}{d^*(e^{\Phi}\varphi)}_{\mu}.
\end{equation}
Moreover, the range $d(H_{k-1})$ is contained in the subspace
$(V_k\cap X_k)'=P'_k V'_k$.
\end{prop}

\begin{rem}
Considering $\Rn$ as a manifold, it would be possible
to use the $C^2$-density furnished by the measure $\mu=e^{-\Phi}\,dx$ (see
e.g.\ \cite[Ch.~6]{H} for the notions), but it is preferable to use the
Lebesgue integral, for this gives an extension of the usual embeddings, such
as $L^2(\mu)\subset L^2_{\op{\loc}}(\Rn)\hookrightarrow \cal D'(\Rn)$.
\end{rem}

\subsection{Continuation of Proofs}
  \label{prf-ssect}
When $R(\A)$ is closed the implication
\eqref{R0-cnd}$\implies$\eqref{eqviX-cnd} of Theorem~\ref{ran-thm}
allows a renorming of $\tilde V:=V\cap X$ such that
$\AA|_{X}\subset\tilde{\cal A}_1$; cf.~$2^\circ$ of
Corollary~\ref{V'-cor}. From Proposition~\ref{dV'-prop} it follows that
$d^*\tilde{\cal A}_1^{-1} d$ is bounded 
\begin{equation}
  L^2(\mu)\longrightarrow (\tilde V)'\longrightarrow
  \tilde V \longrightarrow L^2(\mu),
\end{equation}
and it coincides with $d^*\AA^{-1}d$ in $H^1(\mu)$ by the remark after
\eqref{VHV'-eq} and therefore with 
$u\mapsto u-\int u\,d\mu$ by \eqref{P'-eq}; extension by continuity
gives \eqref{P-eq} for all $u\in L^2(\mu)$. Closure of
\eqref{Brascamp--Lieb-eq} similarly yields \eqref{Brascamp--Lieb-id}:
indeed, for $g_j\in H^1(\mu)$ one can take $v=\AA|_{X_1}^{-1}dg_1$ and
$f=d\overline{g_2}$ in formula \eqref{VHV'-eq} so \eqref{Brascamp--Lieb-eq} (and the
obvious interpretation of gradients as differentials) gives
\begin{equation}
  \begin{split}
  \scal{\AA^{-1}\nabla g_1}{\nabla g_2}_{\mu} &=
  \scal{\AA|_{X_1}^{-1}dg_1}{dg_2}_{\mu}= 
  \dual{d\overline{g_2}}{\tilde{\cal A}_1^{-1}dg_1}_{\tilde V'\times \tilde V}
\\ 
  &=  \scal{Pg_1}{Pg_2}_{\mu}.
  \end{split}
\end{equation}
The last equality extends to all $g_j$ in $L^2(\mu)$ in view of the density
of $H^1(\mu)$ and the continuity of $d\colon L^2(\mu)\to \tilde V'$ and
$\tilde{\cal A}_1^{-1}d\colon L^2(\mu)\to \tilde V$; cf.~Proposition~\ref{dV'-prop}.

\bigskip

For the Brascamp--Lieb inequality \eqref{Brascamp--Lieb-eq'} with $f\in
L^2(\mu)\cap H^1_{\loc}$, \eqref{Brascamp--Lieb-lb} still shows that $\AA>0$
in the uniformly strictly convex case; from
\eqref{m1-cnd}$\implies$\eqref{R0-cnd} of Theorem~\ref{ran-thm} it follows 
that \eqref{Brascamp--Lieb-id} is available for $g_j=f$. From
\eqref{m1-cnd}$\implies$\eqref{eqvi-cnd} we see that 
$1^\circ$ of Corollary~\ref{V'-cor} applies. Since $\cal A_1$ is an isometry,
$a_1(\cal A_1^{-1}\cdot,\cal A_1^{-1}\cdot)$
is an inner product on $V'$ inducing the norm, so for $v\in V'$,
\begin{equation}
  \dual{v}{\overline{\cal A_1^{-1}v}}_{V'\times V}=\norm{v}{V'}^2
  =\sup\bigl\{\,\frac{|\dual{v}{\overline{w}}|^2}{
   a_1(w,w)} \bigm| w\in C^\infty_0(\Rn,\wedge^1\Cn)\setminus\{0\}\,\bigr\}.
\end{equation}
Because of the $H^1_{\loc}$-condition, $(\nabla
f)^T(\Phi'')^{-1}\overline{\nabla f}$ is a well defined function belonging
to $L^1_{\loc}$; if it has finite integral, then $v=df$ is in
$D(\Phi'')\subset L^2(\mu,\wedge^1\Cn)$ so that \eqref{Brascamp--Lieb-lb} may be invoked
as in the argument for \eqref{APhi-eq}, which hence also holds in this case.  

When $\Phi$ is merely strictly convex, the reduction to the uniform case
carries over verbatim.

\section{Criteria for Closed Range}   \label{inv-sect}
Because $Z(\A)$ has finite dimension, the closed-range requirement in
\eqref{A-cnd} is satisfied when $0\notin\sigma_{\ess}(\A)$,
which holds when $\Phi(x)$ is well behaved near $\infty$:
\begin{prop}
  \label{AF-prop}
If $\Phi$ in addition to \eqref{phi-cnd} satisfies
\begin{Rmlist}
  \item   \label{infty-cnd}
$\frac12|\Phi'|^2 \mlap \Phi \ge c>0 $ in a neighbourhood of $\infty$,
\end{Rmlist}
then $0\notin \sigma_{\ess}(\A)$, so \eqref{A-cnd} and the projection identity
\eqref{Brascamp--Lieb-eq} hold. 
\end{prop}

\begin{proof} 
$\A$ is equivalent to
$B_0=\mlap+\tfrac{1}{4}|\Phi'|^2-\tfrac{1}{2}\lap\Phi$, cf.\ 
Section~\ref{unit-ssect}. The latter is essentially self-adjoint by Kato's
result \cite{Kat73}, so the Persson--Agmon formula \cite{Per60},
\cite[Thm.~3.2]{Agm78} may without ambiguity be used for the estimate:
\begin{equation}
  \begin{split}
  \inf \sigma_{\ess}(\A)&=
  \sup_{K\subset\subset\Rn} 
  \inf\bigl\{\,\scal{B_0\varphi
    }{\varphi} \bigm| \varphi\in C^\infty_0(\Rn\setminus K),\quad
   \|\varphi\|=1\,\bigr\}
\\
  &\ge  
  \inf\bigl\{\,\scal{(\tfrac{1}{4}|\Phi'|^2-\tfrac{1}{2}\lap\Phi)\varphi
    }{\varphi} \bigm| \varphi\in C^\infty_0(\Rn\setminus K_0),\
   \|\varphi\|=1\,\bigr\}
\\
  &\ge c/2 >0  
  \end{split}
\end{equation}
when \eqref{infty-cnd} holds in $\Rn\setminus K_0$. 
This yields \eqref{A-cnd}.  
\end{proof}

In addition, \eqref{infty-cnd} combined with a growth condition implies
the stronger fact that $\sigma_{\ess}(\A)=\emptyset$, as shown 
below. Note that whenever $0<\eta<1$ and
$\theta\in\,]\eta,1[$, 
\begin{equation}
  \theta|\Phi'|^2\mlap\Phi=  (\eta|\Phi'|^2\mlap\Phi) +(\theta-\eta)|\Phi'|^2,
\end{equation}
so if $|\Phi'(x)|\to\infty$ for $|x|\to\infty$
and the first term on the right hand side is known to have finite infimum,
consequently the left hand side tends to $\infty$ for $x\to\infty$. Taking
$\eta=1/2$, this shows that \eqref{infty-cnd} together with
$\lim_{|x|\to\infty} |\Phi'(x)|=\infty$ implies condition~\eqref{grw-cnd}
below. Similarly one finds that \eqref{grw-cnd} is more general than
the results one would obtain from the techniques of
J.-M.~Kneib and F.~Mignot in \cite[Lem.~5]{KnMi} (where the proof contains a minor flaw). 

\begin{prop}
  \label{ess-prop}
If $\Phi$ satisfies the following condition
\begin{Rmlist} \setcounter{Rmcount}{1}
  \item  \label{grw-cnd}
    $\exists \theta\in\,]0,1[\colon \lim_{|x|\to\infty}
    \theta|\Phi'(x)|^2\mlap\Phi(x) =\infty$, 
\end{Rmlist}
then $H^1(\mu)\hookrightarrow L^2(\mu)$ is compact, and consequently
$\sigma_{\ess}(\A)=\emptyset$.
\end{prop}

That \eqref{grw-cnd} is sufficient may be proved along the lines of
P.~Bolley, Dauge and Helffer \cite{BoDaHe89}
(even directly, that is without the unitary transformation in
Section~\ref{unit-ssect});
because of this reference's inaccessibility we shall supply the details. 

\begin{proof}
Introducing the vector fields $X_j=\partial_j$ and
their formal adjoints $X^*_j=-\partial_j+\Phi'_j$, one has
when $u\in C^\infty_0(\Rn)$ for their sum and commutator
\begin{equation}
  (X_j+X^*_j)u = \Phi'_j u,\qquad
  [X_j,X^*_j]u = \Phi''_{jj} u.  
\end{equation}
Now it is straightforward to see that 
\begin{align}
  \scal{[X_j,X^*_j]u}{u}_{\mu}&= \nrm{X^*_ju}{\mu}^2- \nrm{X_ju}{\mu}^2  \\
  \nrm{(X_j+X^*_j)u}{\mu}^2&\le
  (1+\tfrac{1}{\varepsilon})\nrm{X_ju}{\mu}^2
  +(1+\varepsilon)\nrm{X^*_ju}{\mu}^2 \quad \forall\varepsilon>0,  
\end{align} 
so that a linear combination of these formulae gives for any $\varepsilon>0$
\begin{equation}
  \scal{(|\Phi'|^2-(1+\varepsilon)\lap\Phi)u}{u}_{\mu}\le
(2+\varepsilon+\varepsilon^{-1}) (\nrm{X_1u}{\mu}^2+\dots+\nrm{X_n u}{\mu}^2).
  \label{comm-eq}
\end{equation}
Because $C^\infty_0(\Rn)$ is dense in $H^1(\mu)$, this inequality is valid
for all $u\in H^1(\mu)$. Indeed, letting
$\mu'=(|\Phi'|^2-(1+\varepsilon)\lap\Phi)\mu$, we infer from 
\eqref{comm-eq} that a fundamental sequence in $H^1(\mu)$ also
converges in $L^2(\Rn,\mu')$, and necessarily to the same limit since
both spaces are embedded into $\cal D'(\Rn)$.

If $u_k\to u$ weakly in $H^1(\mu)$, assumption \eqref{grw-cnd}
implies that $\Psi:=|\Phi'|^2-(1+\varepsilon)\lap\Phi$ is positive in
a neighbourhood of $\infty$ when $\theta=(1+\varepsilon)^{-1}$, so
by \eqref{comm-eq},
\begin{equation}
  \begin{split}
  \int |u_k|^2e^{-\Phi}\,dx &\le \int_{|x|<R} |u_k|^2e^{-\Phi}\,dx +
  \int_{|x|\ge R} \frac{\Psi|u_k|^2}{\inf_{\Rn\setminus B(0,R)}\Psi}e^{-\Phi}\,dx
  \\[2\jot]
  &\le C_\Phi\norm{u_k}{L^2(B(0,R))}^2 +
  \frac{C_\varepsilon \norm{u_k}{H^1(\mu)}^2}{\inf\Set{\Psi(y)}{|y|\ge R}}.
  \end{split}
\end{equation}
Hence \eqref{grw-cnd} and the compactness of
$H^1(B(0,R))\hookrightarrow L^2(B(0,R))$ show that a subsequence of
$u_k$ tends to $0$ in $L^2(\mu)$.

If $\lambda\in\sigma_{\ess}$ there is
$u_k\in D(\A)$ such that $\|u_k\|_{\mu}=1$ while $u_k\to 0$ weakly and
$\|(\A-\lambda)u_k\|_{\mu}\to 0$. Since 
\begin{equation}
  \begin{split}
  \norm{u_k}{H^1(\mu)}&= \norm{u_k}{L^2(\mu)}+ a_0(u_k,u_k)
\\
  &\le 1+\|(\A-\lambda)u_k\|_{\mu}\|u_k\|_{\mu}+|\lambda|\|u_k\|_{\mu}^2
   \le 1+|\lambda| + \cal O(1),
  \end{split}
\end{equation}
the sequence $(u_k)$ is bounded in $H^1(\mu)$, but without convergent
subsequences in $L^2(\mu)$, so the embedding is non-compact. Thus
$\sigma_{\ess}(\A)=\emptyset$ is shown. 
\end{proof}

\section{A Pseudo-differential View Point}
  \label{psd-sect}
As shown in the following, a few extra assumptions on $\Phi(x)$ lead to
domain characterisations,
essential self-adjointness of the $A_j$ and  positivity of $\AA$ (in
addition to closed ranges). 

Actually $C^\infty$-smoothness with a little control of the higher order
derivatives is enough to invoke the calculus in
\cite[18.4--6]{H}, and in this framework $\A$ and $\AA$ are easily
seen to be Fredholm operators if $|\Phi'|$ tends to $\infty$ at
infinity.
Therefore it is assumed in this section that  
\begin{Rmlist}\setcounter{Rmcount}{2}
  \item  \label{Phi8-cnd} $\Phi\in C^\infty(\Rn,\R)$,
  \item  \label{fip-cnd} $|\Phi'(x)|\to\infty$ for $|x|\to\infty$,
  \item  \label{pcp-cnd}  for $|\alpha|\ge1$ there are constants $C_\alpha$
such that  
    \begin{equation}
      |D^\alpha\Phi(x)|\le C_\alpha (1+|\Phi'(x)|^2)^{1/2},
    \end{equation}
  \item  \label{fibnd-cnd} $D^\beta\Phi$ is bounded on $\Rn$ when $\beta$
has a fixed length, say $M\in \N$.
\end{Rmlist} 
This implies that $\Phi(x)$ is slowly increasing, $\Phi\in\cal O_M(\Rn)$, 
so $\Phi$ of, say exponential growth is ruled out; thus the stronger
conclusions of this section have their price.

\subsection{An auxiliary Schr{\"o}dinger operator}
To exploit \eqref{Phi8-cnd}--\eqref{fibnd-cnd} above, we shall henceforth work in the
unweighted space $L^2(\Rn)$ and with the Witten-Laplacians ensuing after
the unitary transformation in Section~\ref{unit-ssect}. That is, we shall consider 
\begin{align}
  \W&=\mlap+\tfrac{1}{4}|\Phi'|^2-\tfrac{1}{2}\lap\Phi
  \label{W0-eq}\\
  \WW&=(\mlap+\tfrac{1}{4}|\Phi'|^2-\tfrac{1}{2}\lap\Phi)
  \otimes I+\Phi'',
  \label{W1-eq}
\end{align}
which act in the distribution sense, and provide them with their maximal domains
in $L^2(\Rn)$ and $L^2(\Rn,\wedge^1\Cn)$, respectively.

For convenience one can here study the auxiliary operator 
\begin{equation}
  P=\mlap+V_0,\quad\text{where}\quad V_0(x)=\tfrac{1}{4}|\Phi'|^2,
\end{equation}
with the domain
\begin{equation}
  D(P)=\Set{u\in H^2(\Rn)}{V_0\cdot u\in L^2(\Rn)}.
  \label{DP-eq}
\end{equation}
To analyse this, let the pseudo-differential operators $p(x,D)$
and $q(x,D)$ have symbols
\begin{equation}
  p(x,\xi)=|\xi|^2+V_0(x),\qquad
q(x,\xi)=((1-\chi(x,\xi))p(x,\xi)+\chi(x,\xi))^{-1} 
\end{equation}
where $\chi\in C^\infty_0(\R^{2n})$ is positive and $\equiv1$ on a
compact set $K$ such that 
\begin{equation}
  (x,\xi)\notin K\implies p(x,\xi)\ge 1.
\end{equation}
This makes $q(x,\xi)$ well defined in $C^\infty(\R^{2n})$.

The calculus in \cite[Ch.~18.4--6]{H} applies to this case with
\begin{equation}
  p(x,\xi)\in S(m^2,g),\qquad q(x,\xi)\in S(m^{-2},g)
\end{equation}
when the weight and metric equal, respectively,
\begin{equation}
  m(x,\xi)=(1+|\xi|^2+|\Phi'(x)|^2)^{1/2},\qquad 
  g=|dx|^2+\frac{|d\xi|^2}{m(x,\xi)^2};
\end{equation}
H{\"o}rmander's notation and terminology is used here and below. 
When applying this theory, condition \eqref{fibnd-cnd} is posed in order to
show that $g$ is $\sigma$-temperate. 

From the calculus we next infer that $q(x,D)$ acts as a parametrix of
$p(x,D)$, i.e.\
\begin{equation}
  p(x,D)q(x,D)=1-K_1(x,D), \qquad q(x,D)p(x,D)=1-K_2(x,D)
  \label{pq-eq}
\end{equation}
where $K_j\in \op{OPS}(m^{-1},g)$. By \cite[Thm.~18.6.6]{H} the $K_j$
are compact in $L^2(\Rn)$ because $m^{-1}$ according to
\eqref{fip-cnd} tends to $0$ at infinity,
for with the choice of $g$ made above we have $g\le g^\sigma$.

To see the latter fact, note that by definition
\begin{equation}
  g^\sigma_{x,\xi}(y,\eta)=
  \sup\Set{\frac{|\dual{\eta}{\hat y}-\dual{y}{\hat
\eta}|^2}{g_{x,\xi}(\hat y,\hat\eta)}}{(\hat y,\hat
\eta)\in\R^{2n}\setminus\{(0,0)\}} 
\end{equation}
so the isometry of the Hilbert space
$(\R^{2n},g_{x,\xi}(\cdot,\cdot))$ onto its dual gives
\begin{equation}
  \begin{split}
  g^\sigma_{x,\xi}(y,\eta)&= \sup_{(\hat y,\hat\eta)\neq 0}
  \left|\frac{g_{x,\xi}((\eta,-m(x,\xi)^2y),(\hat y,\hat\eta))}
           {|(\hat y,\hat\eta)|_{g_{x,\xi}}}\right|^2
  \\
  &=(|(\eta,-m(x,\xi)^2y)|_{g_{x,\xi}})^2
  \\
  &=|\eta|^2+m(x,\xi)^2|y|^2= m(x,\xi)^2g_{x,\xi}(y,\eta).
  \end{split}
\end{equation}

This shows for one thing the claim that $g\le
g^{\sigma}$, because $m\ge1$, and for another that 
\begin{equation}
  h(x,\xi):=\sup(g_{x,\xi}/g^{\sigma}_{x,\xi})^{1/2}=m(x,\xi)^{-1}.
\end{equation}
Using \cite[18.5.10]{H} we can pass to the Weyl calculus and conclude
that
\begin{equation}
  \begin{aligned}
  p(x,D)&=a^W(x,D), &\qquad a(x,\xi)&= e^{-i\dual{D_x}{D_\xi}/2}p(x,\xi)\\
  q(x,D)&=b^W(x,D), &\qquad b(x,\xi)&= e^{-i\dual{D_x}{D_\xi}/2}q(x,\xi)
\end{aligned}
\end{equation}
with the remainder information that, since $h\cdot m^2=m$, 
\begin{align}
  a(x,\xi)&=p(x,\xi)+R_1(p),\quad & a&\in S(m^2,g),\  R_1(p)\in S(m,g)\qquad
\\
  b(x,\xi)&=q(x,\xi)+R_1(q),\quad & b&\in S(m^{-2},g),\  R_1(q)\in S(m^{-3},g)\qquad  .
\end{align}
Moreover, by \cite[18.4.3]{H} and \cite[18.5.4]{H},
\begin{gather}
  R_1(p)q + pR_1(q)\in S(m^{-1},g) \\
  R_1(a,b)\in S(hm^{-2}m^2,g)=S(m^{-1},g).
\end{gather}
This gives finally, by the compact support of $\chi$,
\begin{equation}
  pq=1+\chi (p-1)((1-\chi)p+\chi)^{-1} \in S(m^{-1},g),
\end{equation}
and hence the relations
\begin{equation}
  \begin{split}
  p(x,D)q(x,D)&= (a\# b)^W(x,D)
  \\
  &=(ab + R_1(a,b))^W(x,D)
  \\
  &=(pq+R_1(p)q+pR_1(q)+R_1(a,b))^W(x,\xi)
  \\
  &=1-K_1(x,\xi)
  \end{split}
  \label{pq'-eq}
\end{equation}
with $K_1(x,\xi)$ in $S(m^{-1},g)$ as claimed. Similarly, because
$R_1(b,a)\in S(m^{-1},g)$, one finds
that $q(x,\xi)p(x,\xi)-1\in\op{OPS}(m^{-1},g)$.

Note that in a similar fashion one has:
\begin{prop}
  \label{comp-prop}
Let $g$ be a $\sigma$-temperate metric fulfilling $g\le g^\sigma$ and
$g_{x,\xi}(t,\tau)\equiv g_{x,\xi}(t,-\tau)$, and let $m_1$ and $m_2$ be
$\sigma$, $g$-temperate weights. Then 
\begin{equation}
  \op{OPS}(m_1,g)\cdot \op{OPS}(m_2,g)\subset \op{OPS}(m_1m_2,g).
\end{equation}
\end{prop}
\begin{proof}
One may use \cite[Th.~18.5.10]{H} and the remark thereafter to express
$p(x,D)q(x,D)$ as $(a\#b)^W(x,D)$ and then apply 18.5.4 and 18.5.10. 
\end{proof}

Obviously the maximal domain of $p(x,D)$ as an operator in $L^2(\Rn)$
is the set $D(p)$ consisting of those $u\in L^2$ for which also
$p(x,D)u\in L^2$. However this equals $D(P)$ in \eqref{DP-eq}: for if
$u\in D(p)$ then \eqref{pq-eq} yields, with $f=p(x,D)u$ in~$L^2$,
\begin{equation}
  u-K_2(x,D)u=q(x,D)f
\end{equation}
and by application of $1+K_2(x,D)$
\begin{equation}
  u=(1+K_2(x,D))q(x,D)f+K_2(x,D)^2 u;
\end{equation}
by the proposition both $K_2(x,D)^2$ and $(1+K_2(x,D))q(x,D)$ are
in $\op{OPS}(m^{-2},g)$, so that $D^\alpha u$ and $V_0u$ are in $L^2$ when
$|\alpha|\le2$ by combined application of Proposition~\ref{comp-prop} and
\cite[18.6.3]{H}. Since also
$(\Phi'_{j_1}(x))^{k_1}(\Phi'_{j_2}(x))^{k_2}\xi^\alpha$ is in $S(m^2,g)$ 
when $k_1+k_2+|\alpha|\le2$, it is seen that $(\Phi')^\beta D^\alpha u\in L^2$,
so that we for later use have the more precise result:

\begin{lem}
  \label{Pdm-lem}
Both $D(P)$ and $D(p(x,D)_{\op{max}})$ coincide with the space of
$u$ for which
\begin{equation}
  (\Phi')^{\beta}D^\alpha u \in L^2(\Rn)
\end{equation}
for all $\alpha$ and $\beta\in\N_0^n$ for which $|\alpha|+|\beta|\le 2$; and
$(\sum_{|\alpha+\beta|\le2} |(\Phi')^{\beta}D^\alpha u|^2)^{1/2}$ is
equivalent to the graph norm of $P$. 
\end{lem}

Thereby $P=p(x,D)_{\op{max}}$, so by duality and symmetry of $P$,
\begin{equation}
  p(x,D)_{\op{min}}\supset P^*\supset P \supset p(x,D)_{\op{max}}
\end{equation}
so that $D(P)$ is both the minimal and maximal domain of
$p(x,D)$. Consequently $P$ is essentially self-adjoint on
$C^\infty_0(\Rn)$.

Furthermore $D(P)$ is a Hilbert space in $P$'s graph norm with $P$ and
$q(x,D)$ belonging to $\B(D(P),L^2)$ and $\B(L^2,D(P))$, respectively. For this reason 
$\dim\coker P\le \dim\coker(1-K_1)<\infty$, where the last
inequality is obtained from the compactness of $K_1$, and hence 
$R(P)$ is necessarily closed; cf.~\cite[19.1.1]{H}.

\bigskip

Returning to $\W$, the perturbation $\mlap\Phi(x)$ of $P$
is easily handled; first we show that the domain of $\W$ equals
$D(P)$. Clearly it contains $D(P)$ in view of \eqref{pcp-cnd} (when
this is used to get the elementary estimate $|\lap\Phi|\le
c(1+|\Phi'|^2)^{1/2}\le c(1+|\Phi'|^2)$, which entails $|u\lap\Phi|\le
c(1+|\Phi'|^2)|u|\in L^2$), and if $f:=\W u$ is in $L^2$ for some $u\in
L^2$, one has in $\cal S'(\Rn)$ 
\begin{equation}
  q(x,D)f=(1-K_2(x,D)-\tfrac{1}{2}q(x,D)\lap\Phi(x))u
  \label{DMA-eq}
\end{equation}
where 
\begin{equation}
  K'_2(x,D)=K_2(x,D)+\tfrac{1}{2}q(x,D)\lap\Phi(x)\in\op{OPS}(m^{-1},g). 
  \label{DMA'-eq}
\end{equation}
By application of $1+K'_2(x,D)$
\begin{equation}
  u=(1+K'_2(x,\xi))q(x,D)f+ K'_2(x,D)^2u,
  \label{DMA''-eq}
\end{equation}
so, like for $P$ above, we find that $u$ is in $H^2$ with $|\Phi'|^2u$
in $L^2$. 

For the range we get that 
\begin{equation}
  \W q(x,D)=p(x,D)q(x,D)-\tfrac{1}{2}\lap\Phi(x)q(x,D)
  =1-K'_1(x,D)
  \label{Afrh-eq}
\end{equation}
where $K'_1(x,D)\in \op{OPS}(m^{-1},g)$ since both $K_1(x,D)$ and
$(\lap\Phi)q(x,D)$ are so; the latter fact is by
Proposition~\ref{comp-prop} because $\lap\Phi$ is in 
$S(m,g)$ in view of condition~\eqref{pcp-cnd}. Using \eqref{fip-cnd} as
above $K'_1(x,D)$ is compact in $L^2$, and since $q(x,D)$ maps all of
$L^2(\Rn)$ into
$D(\W)$ by the just shown characterisation of this set, $\W$ has closed range.

Using that $D(\W)=D(P)$, it is straightforward to see that $\W$ is
self-adjoint: for if $u\in D(\W )$, then $d_{\Phi}u$ is in $L^2$ by
\eqref{pcp-cnd} and the just shown characterisation of
$D(\mlap^{(0)}_\Phi)$, and in addition we may by Lemma~\ref{dns-lem}
approximate any $v$ in the domain of $d_{\Phi}$ by
functions $\varphi_k\in C^\infty_0(\Rn)$ and get
\begin{equation}
  \scal{d_{\Phi}u}{d_{\Phi}v}=
  \lim_{k\to\infty} \dual{\mlap u
         +\tfrac{|\Phi'|^2}{4}u-\tfrac{\lap\Phi}{2}u}{\overline{\varphi_k}}
  =\scal{\mlap^{(0)}_\Phi u}{v},
  \label{dom-id}
\end{equation}
so that $\W$ coincides with the self-adjoint Lax--Milgram operator $B_0$
defined in Section~\ref{unit-ssect}.
Because of this it is moreover essentially self-adjoint on $C^\infty_0(\Rn)$.

Altogether we have the following:

\begin{thm}
  \label{psdCR-thm}
Let $\Phi(x)$ have the properties in \eqref{phi-cnd} and
\eqref{Phi8-cnd}--\eqref{fibnd-cnd} above. Then $\A$ has closed range in
$L^2(\Rn,\mu,\C)$, so that the conclusions of Theorem~\ref{Brascamp--Lieb-thm}
are valid, and $\A$ is essentially self-adjoint on $C^\infty_0(\Rn)$.
  
Moreover, $D(\A)$ consists of the functions $u$ for which
$(\Phi')^\beta D^\alpha u$ belongs to $ L^2(\Rn,\mu)$
for all $\alpha$ and $\beta\in\N_0^n$ such that $|\alpha+\beta|\le 2$.
\end{thm}

Observe that the closed range, and even $\sigma_{\ess}(\A)=\emptyset$, is an
immediate consequence of Proposition~\ref{ess-prop} since
\eqref{fip-cnd}--\eqref{pcp-cnd} imply condition~\eqref{grw-cnd} there.

However, the density of $C^\infty_0(\Rn)$ in the graph norm has only been
obtained because the pseudo-differential techniques made an analysis of
the maximal domains possible.

\subsection{Applications to $\boldsymbol{\AA}$}

Using the same line of thought as for $\A$ one finds:

\begin{thm}
  \label{psdAA-thm}
Let $\Phi(x)$ satisfy \eqref{phi-cnd} together with
\eqref{Phi8-cnd}--\eqref{fibnd-cnd} above. Then the domain of $\AA$ is
given by
\begin{equation}
  D(\AA)=\Set{\sum_{j=1}^n v_jdz^j}{\forall j\colon |\alpha+\beta|\le 2
  \implies (\Phi')^{\beta} D^{\alpha}v_j\in L^2(\mu,\Rn)},
  \label{DmA1-eq}
\end{equation}
and $\AA$ is essentially self-adjoint from $C^\infty_0(\Rn,\wedge^1\Cn)$ and
has closed range.
\end{thm}

Indeed, that $D(\WW)$ contains $D(\W)\times\dots\times D(\W)$ (cf.\ the
set in \eqref{DmA1-eq}) is clear by \eqref{W1-eq} and \eqref{pcp-cnd};
if conversely $w:=\WW v$ is in $L^2(\Rn,\wedge^1\Cn)$ for some $v$ there,
the procedure in \eqref{DMA-eq}--\eqref{DMA''-eq} gives, when $q(x,D)$ is
tensored with~$I$,
\begin{equation}
  v=((1+\tilde K_2(x,D))q(x,D))\otimes Iw+\tilde K_2(x,D)^2 v,
  \label{K2t-eq}
\end{equation}
where $\tilde K_2(x,D)$ equals $K_2'(x,D)\otimes I-q(x,D)\otimes I\cdot
\Phi''$ and has all of its entries in $\op{OPS}(m^{-1},g)$. Therefore the
inclusion from the left to the right in \eqref{DmA1-eq} follows.

When applying $q(x,\xi)\otimes I$ as a right-parametrix we find
\begin{equation}
  \WW(q(x,D)\otimes I)= (1-K'_1(x,D))\otimes I +
  \Phi''(q(x,D)\otimes I) =: I-\tilde K_1(x,D)
\end{equation} 
where each entry of $\tilde K_1(x,D)$ is in $\op{OPS}(m^{-1},g)$, and hence
compact in $L^2(\Rn)$. Now $q\otimes I$ sends $L^2(\Rn,\wedge^1\Cn)$ into
$D(\WW)$, so this shows that $\WW$ has closed range; cf.~the argument for $P$
above. 

To show the self-adjointness one can identify $\WW$
with $ U\AA U^*$, where $U=U\otimes I$;
cf.~Section~\ref{unit-ssect}. Now $d_{\Phi}$ and $d_{\Phi}^*$ map any $v\in
D(\WW)$ into $L^2$ by the characterisation of $D(\WW)$, so as in
\eqref{dom-id} one finds that $\WW$ is contained in the self-adjoint
Lax--Milgram operator defined from $(L^2(\Rn,\wedge^1\Cn),V_{\Phi},
b_1)$. Hence $\WW$ equals this as well as the minimal realisation of the
expression in \eqref{W1-eq}. Consequently $\WW$ and $\AA$ are essentially
self-adjoint. 

Note that $D(\AA)\hookrightarrow H^1(\mu)^n$ by Lemma~\ref{Pdm-lem}, so that
$\sigma_{\ess}(\AA)$ is empty. 

\bigskip

Injectivity of $\AA$ may be obtained in the set-up above as soon as
\eqref{fip-cnd} is strengthened to a specific growth rate at infinity; that
is when \eqref{fip-cnd} is replaced~by: 

\begin{itemize} \setlength{\itemindent}{2\jot}
    \item[(IV$_\omega$)] There exist $\omega>0$ and $C>0$ such that
\begin{equation}
  x\cdot\Phi'(x)\ge \tfrac{1}{C} |x|^{1+\omega}\quad\text{for}\quad |x|\ge C.
\end{equation}
\end{itemize}
Since $C|\Phi'|\ge |x|^{\omega}$ holds a fortiori, $\AA$ is then moreover
strictly positive because of the closed range obtained in
Theorem~\ref{psdAA-thm}:

\begin{thm}
  \label{AApos-thm}
Let $\Phi$ satisfy \eqref{phi-cnd}, 
\eqref{Phi8-cnd}--\eqref{fibnd-cnd} and {\rm (IV$_\omega$)}. Then $\AA>0$ on
$L^2(\Rn,\mu,\wedge^1\Cn)$. 
\end{thm}
\begin{proof}
As remarked it suffices to show injectivity of $\WW$, and for this it is
enough that 
\begin{equation}
  Z(\WW)\subset \Set{\sum v_j\,dz_j}{ v_j\in\cal S(\Rn)\ \forall j}
  =: \cal S(\Rn,\C^n).
  \label{ZWW-eq}
\end{equation}
Indeed, given $v\in Z(\WW)\cap\cal S(\Rn,\C^n)$, condition (IV$_\omega$) will
imply that 
\begin{equation}
  f(x)=\int_0^1 e^{-(\Phi(x)-\Phi(tx))/2} x\cdot v(tx)\,dt
  \label{f-id}
\end{equation}
defines an element $f(x)$ of $L^2(\Rn)$ for which $d_\Phi f=v$; since $v\in
Z(\WW)\subset Z(d_\Phi^*)=R(d_\Phi)^\perp$ this will give $v=0$ as desired.

It is straightforward to see that $f\in C^0(\Rn)$ for such $v$, and
\eqref{f-id} gives, cf.~\eqref{U-def}, 
\begin{equation}
  U^*f(x)=\int_0^1 x\cdot U^*v(tx) \,dt
\end{equation}
hence that $dU^*f=U^*v$; therefore $f=e^{-\Phi/2}U^*f$ is in
$C^\infty(\Rn)$ by \eqref{ZWW-eq} and \eqref{Phi8-cnd}. This also yields
$v=UdU^*f=d_\Phi UU^*f= d_\Phi f$ as claimed.

That $f$ is in $L^2(\Rn)$ follows if $|x|^kf(x)$ is bounded for large
$k$. Using (IV$_\omega$) we get, for $|x|>C$ and $t|x|>C$,
\begin{equation}
  \Phi(x)-\Phi(tx)=\int_t^1 x\cdot\Phi'(sx)\,ds
  \ge C^{-1}|x|^{1+\omega}\int_t^1 s^\omega\,ds=
  \tfrac{|x|^{1+\omega}}{C(1+\omega)}(1-t^{1+\omega}),
\end{equation}
so the fact that $(1-t)^{1+\omega}\le 1-t^{1+\omega}$ on $[0,1]$ (since these
functions are convex and concave, respectively, on this interval) leads to
the conclusion that
\begin{equation}
  \Phi(x)-\Phi(tx)\ge \frac{|x-tx|^{1+\omega}}{C(1+\omega)}
  \quad\text{for}\quad |x|,\,|tx|\ge C.
  \label{phi-ineq}
\end{equation}
Next one may for each $k\in\N$ deduce the existence of a constant $C_1$ such
that when $|x|\ge C$, then (with $\ang{x}:=(1+|x|^2)^{\tfrac{1}{2}}$ in the
rest of this proof)
\begin{equation}
  \ang{x-y}^k e^{-(\Phi(x)-\Phi(y))/2} \le C_1 \quad\text{for all}
  \quad y\in \ch\{0,x\}.
\end{equation}
when $\op{ch}A$ denotes the convex hull of $A$.
Indeed, if $|y|\ge C$ the inequality in \eqref{phi-ineq} yields that the
left hand side, when $r:=\ang{x-y}\in\R_+$, is estimated by
\begin{equation}
  (r^\beta e^{-r})^\gamma \quad\text{for some}\quad \beta,\,\gamma>0;
\end{equation} 
when $|y|<C$ one can let $z=\tfrac{C}{|x|}x$ and reduce to the case
$|y|\ge C$, using the inequalities 
\begin{gather}
  \ang{x-y}^k\le\ang{x}^k\le\ang{z}^k\ang{x-z}^k\le(1+C^2)^k\ang{x-z}^k
\\
  \exp(\Phi(y)-\Phi(z))\le \exp(\sup_{|y|\le C}\Phi-\inf_{|z|\le C}\Phi)
  <\infty.  
\end{gather}

For $f(x)$ this now gives, since $|x|<C$ is easy,
\begin{equation}
  \begin{split}
    \sup_{x\in\Rn} |x|^k|f(x)| &\le 
    \sup_{x\in\Rn} \int_0^1\ang{x-tx}^{k+1}e^{-(\Phi(x)-\Phi(tx))/2}\,dt
\\[-2\jot]
    &\qquad\qquad\qquad
    \times \sup\Set{\ang{y}^{k+1}|v(y)|}{y\in\Rn} 
    <\infty. 
  \end{split}
\end{equation}

It remains, therefore, to show \eqref{ZWW-eq}. If $v\in Z(\WW)$, then
$v=\tilde K_2(x,D)v$, hence
\begin{equation}
  v=\tilde K_2(x,D)^N v \quad\text{for every}\quad N\in\N;
\end{equation}
cf.~\eqref{K2t-eq} ff. Since $\tilde K_2(x,D)^N$ by
Proposition~\ref{comp-prop} has entries in $\op{OPS}(m^{-N},g)$, it follows
that $|\Phi'(x)|^{N}D^{\beta} v_j(x)$ is in $L^2(\Rn)$ for all multiindices
$\beta$ and all $j$ and $N$. Moreover, 
\begin{equation}
  \ang{x}^N\le C^{-1}|\Phi'(x)|^{N/\omega}\le C^{-1}\ang{\Phi'(x)}^{N'}\le C'
  (1+|\Phi'(x)|^{N'})  
\end{equation}
when $N'\ge N/\omega$; cf.~(IV$_\omega$). It is thus shown that $v$ has
coefficients in $\cal S(\Rn)$, and altogether this shows the theorem.
\end{proof}

\subsection{Example}  \label{exmp-ssect}
Consider the potential, as done by Helffer \cite{Hel96,Hel97'} and many others,
\begin{equation}
  \Phi(x)=\tfrac{1}{h}\sum_{j=1}^n
  (\tfrac{\lambda}{12}x_j^4+\tfrac{\nu}{2}x_j^2) 
  +\tfrac{1}{h}\tfrac{\cal I}{2}\sum_{j=1}^n 
  |x_j-x_{j+1}|^2,
\end{equation}
whereby $x_{n+1}=x_1$ as a convention. Here $h>0$ and $\cal I>0$ while 
\begin{equation}
  \lambda>0>\nu.
\end{equation}
Therefore $\Phi(x)$ is not convex, so the Brascamp--Lieb inequality does not
apply to this case. The condition $\int e^{-\Phi}\,dx=1$ may be fulfilled by
adding an $h$-dependent constant. Moreover, \eqref{Phi8-cnd} and
\eqref{pcp-cnd} clearly hold. Concerning (IV$_\omega$), Euler's formula gives
\begin{equation}
  x\cdot \Phi'(x)
  \ge \tfrac{1}{h}\sum(\tfrac{\lambda}{3}x_j^4+\nu x_j^2)
  =\tfrac{|x|^4}{h}(\tfrac{\lambda}{3}(\tfrac{1}{n})^2-|\nu||x|^{-2})
 \ge C^{-1}|x|^4
\end{equation}
when $|x|\ge C$ for some sufficiently large $C=C(h,n,\lambda,\nu)$; here it
is used that $|x_j|^2\ge|x|^2/n$ for some $j\in\{1,\dots,n\}$. Consequently
(IV$_\omega$) holds for all $\omega\in\,]0,3]$ for the above $\Phi(x)$. (By
comparison, the assumptions in \cite{Sj96} are unfulfilled since the 
$\Phi''_{jk}$ are unbounded on $\Rn$.)

Because of this, the corresponding operators $\A$ and $\AA$ have the
properties given in Theorems~\ref{Brascamp--Lieb-thm}, \ref{psdCR-thm},
\ref{psdAA-thm} and \ref{AApos-thm}, in particular \eqref{Brascamp--Lieb-id}
holds because $\AA>0$. 

The lower bound  $m(\AA)$ can moreover, for certain $h$ and $\cal I$, be
estimated in various ways, see for example \cite{Hel96,Hel97'}. 

\section{Final Remark}      \label{final-sect} 
The essential self-adjointness of $\A$ and $\AA$
(or $\Ak$) holds in a greater generality than that established in
Section~\ref{psd-sect}. For scalar Schr{\"o}dinger operators this is well
known from works of T.~Kato \cite{Kat73} and S.~Agmon \cite{Agm78}, but
especially C.~Simader's note \cite{Sim78} appears useful for an extension to
`systems' like $\lap^{(k)}_{\Phi}$. 
In fact, Simader's argument for $\mlap+V$ specialised to the case $V\in
C^0(\Rn,\R)$ appears in a recent lecture note \cite[Thm.~9.4.1]{Hel99},
and in this form it is straightforward to carry over to $\mlap\otimes I+V$
with $V\in C^0(\Rn,\R^{n^2})$, when this operator is positive on
$C_0^\infty(\Rn,\Rn)$, hence to $\lap^{(1)}_{\Phi}$ and $\AA$. 

However, the domain characterisations and the corollary on the compact
resolvent (in particular of $\AA$) should in any case motivate the given
applications of the Weyl calculus. 

\appendix

\section{Forms with distributions as coefficients}   \label{curr-app}
The general framework for distribution-valued differential forms, 
so-called currents, is given by G.~de~Rham \cite{Rham55} and
L.~Schwartz~\cite{Swz59}. However, the definition of $E$-valued distributions as
continuous linear maps $\cal D(\Omega)\to E$ given in \cite[Ch.~1 \S2]{Swz59}
leads to severe difficulties (cf.\ the introduction of \cite{Swz59}) in the
proof that $\cal D'(\Omega,E)$ is the dual of $C^\infty_0(\Omega,E')$; for
the finite-dimensional example $E=\wedge^k\Cn$, a much more direct
approach is given in Schwartz' book \cite[Ch.~9]{Swz66} where 
differential forms on manifolds are treated.

In the present article where $\Omega=\Rn$ is a flat, oriented manifold,
further simplifications are given below for the reader's sake. The
definition of $\cal D'(\Rn,\wedge^k\Cn)$ as the dual of
$C^\infty_0(\Rn,\wedge^k\Cn)$ is a little unconventional (testfunctions
valued in $\wedge^{n-k}\Cn$ is common), but this choice is consistent with the
made identification of $L^2(\Rn,\mu,\wedge^k\Cn)$ and its dual.

For precision,  $C^\infty_0(\Rn,\wedge^k\Cn)$ denotes the compactly
supported, infinitely differentiable maps $\Rn\to\wedge^k\Cn$ (i.e.\ into
the space of anti-symmetric $k$-linear forms on $\C$). The canonical coordinates
$z_1$,\dots,$z_n$ in $\Cn$ lead to a basis for
$\wedge^k\Cn$ consisting of $dz^J:=dz_{j_1}\wedge\dots\wedge dz_{j_k}$,
where $J=(j_1,\dots,j_k)$ is an increasing $k$-tuple. 
Therefore any $\varphi\in C^\infty_0(\Rn,\wedge^k\Cn)$
equals $\sum'\varphi_J\,dz^{J}$ with unique $\varphi_J\in
C^\infty_0(\Rn)$. Thus there is a bijection
\begin{equation}
  \cal J\colon C^\infty_0(\Rn,\wedge^k\Cn)\to 
  {\textstyle \prod_{1\le j\le\binom{n}{k}}} C^\infty_0(\Rn);
  \label{J-eq}
\end{equation}
and there is a unique topology on the domain which makes $\cal J$ a
homeomorphism, when the codomain has the product topology. (For brevity,
indexation on $\prod$ is suppressed below.)

The dual of $\cal J$'s codomain is isomorphic to $\prod \cal D'(\Rn)$, for
any continuous linear functional $F$ acts on $\varphi=(\varphi_j)$ as
$F(\varphi)= \sum F\circ I_j(\varphi_j)$ where $I_j$ sends $\varphi_j$ into
$(0,\dots,\varphi_j,\dots,0)$; since $F\circ I_j$ is continuous it is in
$\cal D'(\Rn)$. 

Now $\cal D'(\Rn,\wedge^k\Cn)$ may be defined as the dual of
$C^\infty_0(\Rn,\wedge^k\Cn)$; equipping dual spaces with their
$\op{w}^*$-topologies, there is by transposition a linear homeomorphism
\begin{equation}
  \cal D'(\Rn,\wedge^k\Cn) \xleftarrow{\;\,\cal J'\,}
  \prod \cal D'(\Rn),
  \label{J'-eq}
\end{equation}
so for $(u_J)$ in $\prod\cal D'(\Rn)$ and
$\varphi=\sum{}'\varphi_J\,dz^J$ in $C^\infty_0(\Rn\wedge^k\Cn)$,
\begin{equation}
  \dual{\cal J'(u_J)}{\varphi}=\dual{(u_J)}{\cal J\varphi}=
  \sum{}'\dual{u_J}{\varphi_J}.
  \label{JJ'-eq}
\end{equation}
Indeed, $\cal J'$ is surjective because any $u\in \cal
D'(\Rn,\wedge^k\Cn)$ gives rise to the continuous linear functional
$u\circ\cal J^{-1}$, which is in $\prod\cal D'(\Rn)$,
in view of \eqref{J-eq} ff, so that for some $(u_J)$ it holds for all
$(\varphi_J)$ in $\prod C^\infty_0(\Rn)$ that
\begin{equation}
  \dual{u}{\cal J^{-1}(\varphi_J)}=u\circ\cal J^{-1}(\varphi_J)
  =\dual{(u_J)}{(\varphi_J)}=\dual{\cal J' (u_J)}{\cal J^{-1}(\varphi_J)}.
\end{equation}
Therefore $u=\cal J'(u_J)$; the rest of \eqref{J'-eq} is straightforward.

It is natural to write $u=\sum{}'u_J\,dz^J$ instead of $u=\cal J'(u_J)$, and
thereby \eqref{JJ'-eq} attains the following more intuitive form,
\begin{equation}
  \dual{u}{\varphi}=\dual{\sum{}'u_J\,dz^J}{\sum{}' \varphi_J\,dz^J}
  =\sum{}' \dual{u_J}{\varphi_J}.
  \label{JJ'-id}
\end{equation}

The usual denseness of $C^\infty_0$ in $\cal D'$ carries over to the
$k$-form-valued spaces by \eqref{J-eq} and \eqref{J'-eq}, and therefore
multiplication, $M_\psi$, by $\psi\in C^\infty(\Rn)$ and the operators
$\partial_j$ extend in a unique way to $\cal D'(\Rn,\wedge^k\Cn)$ as
usual. More precisely, $M_\psi$ and $\partial_j$ are both continuous on
$C^\infty_0(\Rn,\wedge^k\Cn)$ because their definitions show that they
act on each $\varphi_J$ (i.e.\ they commute with $\cal J$), so the
transposed operators $(M_\psi)'$ and $\partial_j'$ act as $M_\psi$ and
$-\partial_j$ on each $u_J$ in \eqref{JJ'-id}, respectively, and they
are therefore denoted by the latter symbols throughout.

In this way one finds that $\cal D'(\Rn,\wedge^k\Cn)$ is a
$C^\infty(\Rn)$-module and that any differential operator $P(\partial)$ with
coefficients in $C^\infty(\Rn)$ is well defined by its action on each $u_J$
(and independent of the canonical choice of $dz^J$); transposition moreover
follows the usual rule. In particular this hold for the exterior derivative
$d_k$, and for this the identies
\begin{gather}
  d^2:= d_{k+1}\circ d_k\equiv 0 \quad\text{on}\quad \cal
  D'(\Rn,\wedge^k\Cn),
  \\
  d(\sum{}'f_J\,dz^J) = \sum_{J}{}'\sum_{j=1}^n \partial_jf_J\, dz^j\wedge dz^J 
  \label{dd-id}
\end{gather}
are obtained by transposition and by closure from
$C^\infty_0(\Rn,\wedge^k\Cn)$, respectively.

%
\baselineskip=0.8\baselineskip

\providecommand{\bysame}{\leavevmode\hbox to3em{\hrulefill}\thinspace}

\end{document}